\documentclass[a4paper,UKenglish,cleveref, autoref]{lipics-v2019}
\title{Finding the Mode of a Kernel Density Estimate} %TODO Please add

\author{Jasper C.H.~Lee}{Brown University}{jasperchlee@brown.edu}{}{}

\author{Jerry Li}{Microsoft Research}{jerrl@microsoft.com}{}{}

\author{Christopher Musco}{New York University}{cmusco@nyu.edu}{}{}

\author{Jeff M. Phillips}{University of Utah}{jeffp@cs.utah.edu}{}{}%TODO mandatory, please use full name; only 1 author per \author macro; first two parameters are mandatory, other parameters can be empty. Please provide at least the name of the affiliation and the country. The full address is optional

\author{Wai Ming Tai}{University of Utah}{wmtai@cs.utah.edu}{}{}

\authorrunning{J.C.H.~Lee \etal}%TODO mandatory. First: Use abbreviated first/middle names. Second (only in severe cases): Use first author plus 'et al.'

\Copyright{Jasper C.H.~Lee, Jerry Li, Christopher Musco, Jeff M. Phillips and Wai Ming Tai}%TODO mandatory, please use full first names. LIPIcs license is "CC-BY";  http://creativecommons.org/licenses/by/3.0/

%TODO mandatory: Please choose ACM 2012 classifications from https://dl.acm.org/ccs/ccs_flat.cfm 

%TODO mandatory; please add comma-separated list of keywords

\category{}%optional, e.g. invited paper

\relatedversion{}%optional, e.g. full version hosted on arXiv, HAL, or other respository/website
\supplement{}%optional, e.g. related research data, source code, ... hosted on a repository like zenodo, figshare, GitHub, ...

\acknowledgements{}%optional

\nolinenumbers %uncomment to disable line numbering

\hideLIPIcs  %uncomment to remove references to LIPIcs series (logo, DOI, ...), e.g. when preparing a pre-final version to be uploaded to arXiv or another public repository

\EventEditors{John Q. Open and Joan R. Access}
\EventNoEds{2}
\EventLongTitle{42nd Conference on Very Important Topics (CVIT 2016)}
\EventShortTitle{CVIT 2016}
\EventAcronym{CVIT}
\EventYear{2016}
\EventDate{December 24--27, 2016}
\EventLocation{Little Whinging, United Kingdom}
\EventLogo{}
\SeriesVolume{42}
\ArticleNo{23}
\usepackage{microtype}%if unwanted, comment out or use option "draft"

\usepackage{amsmath}
\usepackage{amsfonts}
\usepackage{amssymb, amsthm, epsfig}
\usepackage{mathrsfs}
\usepackage{bbold}

\usepackage{graphicx}
\usepackage{verbatim}
\usepackage{color}
\usepackage{epstopdf}

\usepackage{mathtools}
\usepackage{float}
\usepackage{algorithm}
\usepackage[noend]{algorithmic}

\usepackage{xspace}
\usepackage{color}

\usepackage{wrapfig}

\usepackage{url}
\usepackage{comment}  

\usepackage{complexity}

\newcommand{\eps}{\varepsilon}

\renewcommand{\R}{\ensuremath{\mathbb{R}}}

\newcommand{\etal}{\emph{et al.}\xspace}

\newcommand{\abs}[1]{\left| #1 \right|}
\newcommand{\setdef}[2]{\left\{ #1 \mid #2 \right\}}

\newcommand{\jasper}[1]{{\color{magenta}$\langle$#1$\rangle$}}

\DeclareMathOperator{\argmax}{argmax}

\newcommand{\norm}[1]{\|#1\|}

\renewcommand{\G}{\mathcal{G}}
\newcommand{\Gr}{\textsf{Grid}}

\newcommand{\Gb}{\overline{\mathcal{G}}}

\newtheorem{observation}[theorem]{Observation}

\begin{document}
	
%\title{Finding Maximum Point of Gaussian Kernel Density Estimation}
%\author{Jasper C.H.~Lee, Jerry Li, Christopher Musco, Jeff M.~Phillips, Wai Ming Tai}
\maketitle
	
\begin{abstract}
Given points $p_1, …, p_n$ in $\R^d$, how do we find a point $x$ which maximizes $\frac{1}{n} \sum_{i=1}^n e^{-\norm{p_i - x}^2}$? In other words, how do we find the maximizing point, or \emph{mode} of a Gaussian kernel density estimation (KDE) centered at $p_1, …, p_n$? Given the power of KDEs in representing probability distributions and other continuous functions, the basic mode finding problem is widely applicable. However, it is poorly understood algorithmically. Few provable algorithms are known, so practitioners rely on heuristics like the ``mean-shift'' algorithm, which are not guaranteed to find a global optimum.

We address this challenge by providing fast and provably accurate approximation algorithms for mode finding in both the low and high dimensional settings. For low dimension $d$, our main contribution is to reduce the mode finding problem to a solving a small number of systems of polynomial inequalities. For high dimension $d$, we prove the first dimensionality reduction result for KDE mode finding, which allows for reduction to the low dimensional case. Our result leverages Johnson-Lindenstrauss random projection, Kirszbraun's classic extension theorem, and perhaps surprisingly, the mean-shift heuristic for mode finding.  

By combining our methods with known coreset results for KDEs, we obtain running times of $O\left(nd + \frac{d}{\eps^2\rho} (\log \frac{d}{\eps\rho})^{O(d)}\right)$ and $O\left(nd + (\log \frac{1}{\eps\rho})^{O(\frac{1}{\eps^2} \log^3\frac{1}{\eps\rho})} \right)$ to compute a $(1-\epsilon)$ approximate mode, where $\rho$ is a lower bound on the density of the true mode, which is always at least $1/n$. For $d=2$ we also present a combinatorial algorithm that improves this result to $O(n + \frac{1}{\eps^2 \rho} \log^2 \frac{1}{\eps \rho})$.
Finally, we note that our main techniques are general, and would easily extend to KDE basis functions beyond Gaussian.
\end{abstract}

\section{Introduction}
Given a point set $P$ in $\mathbb{R}^d$ and a kernel $K:\mathbb{R}^d\times \mathbb{R}^d\rightarrow \mathbb{R}$, the kernel density estimation (KDE) is a function mapping from $\mathbb{R}^d$ to $\mathbb{R}$ and is defined as $\frac{1}{\abs{P}}\sum_{p\in P}K(x,p)$ for any $x\in \mathbb{R}^d$.
One common example of kernel $K$ is the Gaussian kernel, $K(x,y) = e^{-\norm{x-y}^2}$ for any $x,y\in\mathbb{R}^d$, which is the focus of this paper.  

These kernel density estimates are a fundamental tool in statistics~\cite{Sil86,Sco92,DG84,DL01} and machine learning~\cite{SS02,GBRSS12,MFSS17}.  
For $d=1$, KDEs with a triangular kernel ($K(x,p) = \max(0,1-|x-p|)$) can be seen as the average over all shifts of a fix-width histogram.  And unlike histograms these generalize naturally to a higher dimensions as a stable way to create a continuous function to represent the measure of a finite point set.  
Indeed, the KDEs constructed on an iid sample from any tame distribution will converge to that distribution in the limit as the sample size grows~\cite{Sil86,Sco92}.  
Not surprisingly they are also central objects in Bayesian data analysis~\cite{JL95,FSG13}.  
Using Gaussian kernels (and other positive definite kernels), KDEs are members of a reproducing kernel Hilbert space~\cite{Wah99,hilbert,MFSS17} where for instance they induce a natural distance between distributions~\cite{smola,joshi2011comparing}.  
Their other applications includes outlier detection \cite{zou2014unsupervised}, clustering \cite{rinaldo2010generalized}, topological data analysis \cite{PWZ15,CFLMRW14}, spatial anomaly detection~\cite{agarwal2013union,HMP19}, and statistical hypothesis testing~\cite{GBRSS12}.

In this paper, we study how to find the maximum point of Gaussian KDE -- the mode of the distribution implied by an input point set $P \subset \mathbb{R}^d$.
It is known that Gaussian KDEs can have complex structure of local maximum~\cite{EFR12,JZBWJ16}, but other than some heuristic approaches~\cite{Carreira-Perpinan:2000,Carreira-Perpinan:2007,zheng2015error,GHP97} there has been very little prior work~\cite{PWZ15,agarwal2013union} (which we discuss shortly) in developing and formally analyzing algorithms to find this maximum.  
Beyond being a key descriptor (the mode) of one of the most common representations of a continuous distribution, finding the maximum of a KDE has many other specific applications.  
It is a necessary step to create a simplicial complex to approximate superlevel sets of KDEs~\cite{PWZ15}; to localize and track objects~\cite{CA05,SBH07}; to quantify multi-modality of distributions~\cite{Sil81}; to finding typical objects, including curves~\cite{GHP97}.

\subparagraph{Problem Definition.}
For any $x,y\in \mathbb{R}^d$, we define the Gaussian kernel as $K(x,y)=e^{-\norm{x-y}^2}$.
The Gaussian kernel density estimation $\Gb_Q(x)$ of a point set $Q$ is defined as $\Gb_Q(x) = \frac{1}{|Q|} \sum_{p \in Q} K(p,x)$, for $x \in \mathbb{R}^d$.
We will sometimes use the notation $\G_Q(x) = |Q| \cdot \Gb_Q(x)$ to simplify calculations.
%$\mathcal{G}_P(x)$ as $\sum_{p\in P} K(p,x)$ for any $x\in\mathbb{R}^d$.
%Suppose $x^*$ be $\argmax_{x\in\mathbb{R}^d} \Gb_P(x)$.\jasper{Is there a reason why it's defined as an arg sup instead of an arg max? Don't we have a continuous function with a maximum within a finite ball?}
%Already with $3$ non-symmetric points in $1$ dimension, it is not clear how to analytically solve for the maximum of the associated Gaussian KDE, so we focus on approximate solutions.
In line with other works on optimization, we focus on the approximate version of the mode finding problem, defined as follows.
Given a point set $P$ of size $n$ where $\max_{x \in \mathbb{R}^d} \Gb_P(x) \ge \rho$ for some parameter $\rho$, and an error parameter $\eps > 0$, the goal is to find $x'$ such that $\Gb_P(x')\geq (1-\eps)\max_{x \in \mathbb{R}^d} \Gb_P(x)$.
We assume that the lower bound $\rho$ is known to the algorithm.
In practice, one should expect that $\rho \ll \eps$, so we aim for algorithms with far smaller dependence on $1/\rho$ than on $1/\eps$; note we can always remove the dependence on $\rho$ be setting $\rho = 1/n$.

%In practice, one should expect that $\rho\ll\eps$, so we aim for algorithms with poly-logarithmic complexity in $1/\rho$, and polynomial complexity in $1\eps$.
%In arbitrarily high dimensions, to completely remove the dependence on $d$ (except for reading the data) we relax these goals, but still emphasize smaller dependence on $\rho$.

\subparagraph{Known Results.}
One trivial approach is exhaustive search.
It is easy to see that the optimal point $x^*$ cannot be too far away from the input data.
More precisely, $x^*$ should be within the radius of $\sqrt{\log\frac{1}{\rho}}$ of a point $p$ for some $p\in P$. We will also formally show this observation in our main proof.
Given the above observation, one can construct a grid of width $\frac{1}{\eps\rho}$ around each point of input data and evaluate the value of $\G_P$ at each grid point.
This approach will allow us to output a solution with additive error at most $\eps n\rho$.
However, the size of the search space could be as large as $O\left(n \left(\sqrt{\log\frac{1}{\rho}}\bigg/\eps\rho\right)^d \right)$ which is infeasible in practice.
A similar approach is suggested by Phillips \etal~\cite{PWZ15}.

Another approach, proposed by Agarwal et al.~ \cite{agarwal2013union}, is to compute the depth in an arrangement of a set of geometric objects.
Namely, it is to find the point that maximizes the number of objects including that point.
Even though their approach generalizes to high dimensions, the authors only analyze the two-dimensional case.
%This result is only analyzed in the two dimensional case.
%However, their setting is more general:
They consider a set $\mathcal{S}$ of segment in $\R^2$ and, for any $x\in\R^2$ and $s\in\mathcal{S}$, define $K(x,s) = K(x,y)$ where $y$ is the closest point on $s$ to $x$.
In order to understand how this setting translate into ours, one should think of $p\in P$ as a degenerated segment.
By discretizing the continuous function $K$ into the level set of it, one can view the problem as computing the depth in an arrangement of a collection of level set.
This approach has a running time of $O(\frac{n}{\eps^4}\log^3n)$ (this implicitly sets $\rho = 1/n$).
%\jasper{I'm confused by this paragraph. Is the approach applicable only to $d=2$ or also to higher dimensions? What does ``their setting is more general" mean, if not?}
%It may apply to higher dimensions, but is only analyzed in d=2 case -- Jeff

\subparagraph{Related Work.}
As mentioned before, computing depth in an arrangement of a set of geometric object is highly related to our problem.
Given a collection $\mathcal{C}$ of geometric object in $\R^d$, one can expressed the depth as $\sum_{c\in \mathcal{C}}\mathbb{1}_c(x)$ where $\mathbb{1}_c(x)$ is the indicator function of $x\in c$.
It is easy to see that KDE is basically the same formula by replacing $\mathbb{1}_c(x)$ with $K(x,p)$.
Namely, one can view finding maximum point of KDE as computing ``fractional'' depth of kernel.
Surprisingly, there are not many non-trivial algorithmic results on computing the depth of high-dimensional geometric objects.  
In general, whenever $\mathcal{C}$ is a collection of bounded complexity (e.g.~VC dimension~\cite{VC71}) objects, the arrangement is always of complexity $O(n^{O(d)})$ and it can be constructed, with the depth encoded, in as much time.
%\jasper{I'm not familiar with the literature here, but is the task computing the depth given a set of arranged geometric objects? What is ``it" in ``it can be constructed"?}
%\jeff{An \emph{arrangement} of a set of objects (e.g., disks) is a decomposition of $\b{R}^d$ into the different regions defined by intersections of those objects.  That is it catalogs all vertices formed by the intersection of two disks, the ``edges'' which are boundary arcs between them, the ``faces'' are contiguous sets of points all contained in the same objects, and so on in higher dimensions.  Once you know what and where all of these regions are, it is trivial to annotate them with how many objects occupy a region -- its depth.  It is a central object and idea in computational geometry (see \url{https://en.wikipedia.org/wiki/Arrangement_(space_partition)}) and we can expect a SOCG reviewer to be familiar with the term and concept.}
For instance, a celebrated improvement is when $\mathcal{C}$ is a collection of axis parallel box, the point of maximum depth can be found in $O(n^{d/2 - o(1)})$ time~\cite{chan2013klee}.
In these cases, for our task we can run such approaches on a sample of size $n_0 = O(\frac{d}{\eps^2 \rho}\log \frac{1}{\rho})$~\cite{LLS01,har2011relative}, so the runtimes still have a $1/\rho^{O(d)}$ term.

\subparagraph{Our Approach and Result.}
We present an approximation scheme that reads the data (to sample it) in $O(nd)$ time, and then its runtime depends only on $1/\eps$ and $1/\rho$.  
%For low dimensions this is $(\frac{1}{\eps^2 \rho}) (\log \frac{1}{\eps \rho})^{O(d)}$ and for high dimensions this is improved further to $(\log \frac{1}{\eps \rho})^{O(\frac{1}{\eps^2} \log^3 \frac{1}{\eps \rho})}$.  
%One may set the relative error parameter $\eps$ to a constant, and observe the mode of $\Gb$ must be at least $1/n$ and set $\rho = 1/n$.  In this case the runtimes become $n (\log n)^{O(d)}$ in low dimensions, and $(\log n)^{O(\log^3 n)}$ in high dimensions.  
%\jasper{that's true right?}
%\jeff{I think this is not the main point.  Perhaps $\rho$ is artificial since we can set $\rho = 1/n$ -- but then there is dependence on $n$.  In my view, the real contribution is either $nd + (\log 1/\rho)^{O(d)})$ or $nd + \textrm{sub-linear}(1/\rho)$ with no real dependence on $d$).}  
At the heart of our algorithm are two techniques: dimensionality reduction and polynomial system solving.
We also use standard coreset results for Gaussian kernel density estimates.

For dimensionality reduction (Section~\ref{sec:JL+KDS}), we use Johnson-Lindenstrauss matrices to project the point set down to low dimensions, and solve the problem in low dimensions.
The crucial issue is, if we solve the mode finding problem in the low dimensional space, it is not immediately clear that the original high dimensional space also has a point that gives a high KDE value.
We resolve this with an application of Kirszbraun's extension theorem~\cite{Kirszbraun:1934,Valentine:1945}, which shows the existence of such a high dimensional point.
To find the actual point in the high dimensional space, we use one step of the \emph{mean-shift} algorithm~\cite{Carreira-Perpinan:2000, Carreira-Perpinan:2007}, which is a known heuristic for the KDE maximum finding problem with provable monotonicity properties. We note that we could have alternatively combined a \emph{terminal dimensionality reduction} result \cite{Narayanan:2019} our mean-shift recovery strategy. Doing so would have given the same level of dimensionality reduction, at the expense of reduced simplicity and runtime efficiency.

In low dimensions, we consider Taylor series truncations of the Gaussian kernel, and reduce the mode finding problem to solving systems of polynomial inequalities (Section~\ref{sec:d-small}).
The result of Renegar~\cite{renegar1992computational} implies that one can find a solution to a system of $\lambda$ polynomial inequalities with degree $D$ and $k$ variables in time $O((\lambda D)^{O(k)})$.
Here, $k$ will essentially be the dimensionality $d$ of the problem, and $\lambda$ will be a constant as shown in our constructions.
%Depending on the (remaining) dimensionality of the problem, we propose two ways of formulating the system of inequalities, one with runtime that is dimension-dependent (and suitable for, say, constant dimensional problems), and another with runtime that is dimension-independent but is exponential in $\poly(1/\eps, \log 1/\rho)$.
%The key difference in the two approaches lies in the choice of where to center the Taylor expansions.
We observe that since the optimal point must be close to one of the points in the input, we can consider a sufficiently fine grid in the vicinity of each input point, which totals to $O(n2^{O(d)})$ grid points.
For each grid point, we formulate and solve a system of polynomial inequalities based on Taylor expansions, up to $O(\log\frac{1}{\rho})$ terms around that grid point.
This gives a running time of $O(n(\log\frac{1}{\rho})^{O(d)})$, where $n$ is the size of the input point set.
%\waiming{only one approach now}For the dimensionless approach, we instead consider Taylor expansions directly around each input data point, meaning that we consider $O(n)$ polynomial systems instead of $O(nO(1)^d)$ many systems.
%This gives a runtime of $O(n^2(C'\log\frac{1}{\rho})^d)$, where $C' < C$.
%If $d$ is high (for example $\poly(1/\eps, \log\frac{1}{\rho})$ but more than a constant), then the gain from having $C' < C$ will significantly offset the difference between $n^2$ and $n$ in the complexities.
%\jasper{need to clean up the polynomial system solving descriptions}

Combining the above ideas with standard coreset results yields approximation schemes for the KDE mode finding problem.
We present two such schemes, one with exponential runtime dependence on the dimensionality $d$ which is more suitable for low dimensions, and another with only linear dependence in $d$ (which is necessary for reading the input) and is designed for the high dimensional regime.
Guarantees of these approximations are captured by Theorems~\ref{thm:low} and~\ref{thm:high} respectively, which we prove in Section~\ref{sec:finalschemes}.

\begin{theorem}[Low dimensional regime]
	\label{thm:low}
	Given $\eps,\rho>0$ and a point set $P\subset\mathbb{R}^d$ of size $n$, we can find $x'\in\mathbb{R}^d$ so
	\[
	\Gb_P(x') \geq (1-\eps)\Gb_P(x^*)
	\]
	if $\Gb_P(x^*)\geq \rho$ where $x^*=\argmax_{x\in\mathbb{R}^d} \Gb_P(x)$ in $O\left(nd+\frac{d}{\eps^2\rho} \cdot \log\frac{1}{\rho\delta} \cdot \left(\log \frac{d}{\eps\rho}\right)^{O(d)}\right)$ time with probability $1-\delta$.
\end{theorem}

\begin{theorem}[High dimensional regime]
	\label{thm:high}
	Given $\eps,\rho>0$ and a point set $P\subset\mathbb{R}^d$ of size $n$, we can find $x'\in\mathbb{R}^d$ so 
	\[
	\Gb_P(x') \geq (1-\eps)\Gb_P(x^*)
	\]
	if $\Gb_P(x^*)\geq \rho$, with probability at least $1-\delta$, where $x^*=\argmax_{x\in\mathbb{R}^d} \Gb_P(x)$ in time
	$O\left(nd + \left(\log \frac{1}{\eps \rho } \right)^{O(\frac{1}{\eps^2}\log^3\frac{1}{\eps\rho})} \cdot \log\frac{1}{\delta} + 
	            \min\{nd \log \frac{1}{\delta}, \frac{d}{\eps^2\rho^2}\log^2\frac{1}{\delta}\}\right)$.
\end{theorem}

One may set the relative error parameter $\eps$, and failure probability $\delta$ to constants, and observe the mode of $\Gb_P$ must be at least $1/n$ and set $\rho = 1/n$.  Then the runtimes become $O(n (\log n)^{O(d)})$ for constant dimensions, and $O(nd + (\log n)^{O(\log^3 n)})$ in high dimensions.

In addition to our result in Theorems~\ref{thm:low} and~\ref{thm:high}, we also consider the special case where $d = 2$.
We present a combinatorial algorithm for the 2-dimensional regime which is undeniably easier to implement.
Here, we borrow the idea from \cite{agarwal2013union} which is to compute the depth.
Instead of simply considering the level sets of the Gaussian kernel (which are circles in our setting), we consider a more involved decomposition.
One important property of the Gaussian kernel is its multiplicatively separability -- namely, the Gaussian kernel can be decomposed into factors, with one factor for each dimension.
We now discretize each factor into level sets (which are simply intervals) and then consider their Cartesian products, generating a collection of axis-parallel rectangles.
A similar idea was also suggested in \cite{phillips2018improved}.
Finally, if we compute the depth of this collection of axis-parallel rectangles, we can find out an approximate mode in time $O(\frac{1}{\eps^2\rho}\log^2\frac{1}{\rho})$.
Note that this approach also works on higher dimensional case, but it would yield a slower running time than our general approaches in Theorems~\ref{thm:low} and~\ref{thm:high}.
The formal guarantees of this 2-d algorithm is captured in Theorem~\ref{thm:2d}.

\begin{theorem}[2-dimensional setting]
	\label{thm:2d}
	Given $\eps,\rho>0$ and a point set $P\subset\mathbb{R}^2$ of size $n$, we can find $x'\in\mathbb{R}^2$ so
	\[
	\Gb_P(x')\geq (1-\eps)\Gb_P(x^*)
	\]
	if $\Gb_P(x^*)\geq \rho$ where $x^*=\argmax_{x\in\mathbb{R}^d} \Gb_P(x)$ in $O\left(n + \frac{1}{\eps^2\rho}(\log \frac{1}{\rho}+\log \frac{1}{\delta})\log (\frac{1}{\eps\rho}\log \frac{1}{\delta})\right)$ time with probability at least $1-\delta$.
\end{theorem}

\section{KDE mode finding via system of polynomial}
\label{sec:d-small}

In this section we provide algorithms that approximately find the maximum of the Gaussian KDE in $\R^d$. %where $d$ is a fixed constant.  
We first denote the following notation.
For a point $p\in \R^d$ and $r>0$, denote $B_{p}(r)$ be $\setdef{y\in\mathbb{R}^d}{\norm{y-p}\leq r}$, namely the Euclidean ball around $p$.
For a point set $P\subset \R^d$ and $r>0$, denote $B_P(r)$ be $\cup_{p\in P} B_{p}(r)$.
For a point set $P\subset \R^d$, a point $q\in \R^d$ and $r>0$, denote $Q_{P,q}(r) = P\cap B_q(r\sqrt{\log \frac{1}{\eps\rho}})$.
Also define an infinite grid $\Gr(\gamma)$ be $\setdef{x=(i_1\gamma, i_2\gamma, \dots,i_d\gamma)}{\text{$i_1,i_2,\dots,i_d$ are integers}}$, parametrized by an cell length $\gamma>0$.

We first make an observation that the maximum must be close to one of the data points, captured by Observation~\ref{lem:union}.
\begin{observation}\label{lem:union}
	$x^*\in B_P(\sqrt{\log \frac{1}{\rho}})$. 
	Recall that $x^*=\argmax_{x\in\mathbb{R}^d} \G_P(x)$. 
\end{observation}

\begin{proof}
	
	Suppose $x^*\notin B_P(\sqrt{\log \frac{1}{\rho}})$.
	Then,
	\begin{align*}
	\G_P(x^*)
	= 
	\sum_{p\in P}e^{-\norm{p-x^*}^2}
	<
	\sum_{p\in P}\rho
	= 
	n\rho
	\end{align*}
	However, $\G_P(x^*)\geq n\rho$ by assumption.	
\end{proof}

The algorithm presented in this section rely crucially on the result of Renegar for solving systems of polynomial inequalities, as stated in the following lemma.
\begin{lemma}[\cite{renegar1992computational}]\label{lem:sys_poly}
	Given $\lambda$ polynomial inequalities with degree $D$ and $k$ variables.
	There is an algorithm either finds a solution that satisfies all $\lambda$ polynomial inequalities or returns NO SOLUTION in $O((\lambda D)^{O(k)})$ time.
\end{lemma}

Before we give details of our algorithm, we present the family of systems of polynomial inequalities we formulate for mode finding.
Let $\textsf{SysPoly}(P,q,r,r',\beta)$ be the following system of polynomial. 
\[
	\sum_{p\in Q_{P,q}(r')} \prod_{i=1}^d\left( \sum_{j=0}^{s-1} \frac{1}{j!}\left(-(x_i-p_i)^2\right)^j \right)  \geq \beta \quad \bigwedge \quad \norm{x - q}^2 \leq r^2\log \frac{1}{\eps\rho}
\]
where $s= (r+r')^2e^2\log \frac{d}{\eps\rho}$.
Also, let $\textsf{SysPoly}(P,q,r,r')$ be the algorithm that performs binary search on $\beta$ of $\textsf{SysPoly}(P,q,r,r',\beta)$ and terminates when the search gap is less than $\frac{1}{10}\abs{P}\eps\rho$.
Note that $\beta$ lies between $0$ and $O(\abs{P})$ which means we need $O\left(\log \left(\abs{P}/\frac{1}{10}\abs{P}\eps\rho\right) \right) = O\left(\log \frac{1}{\eps\rho}\right)$ iterations in binary search.
The total running time of $\textsf{SysPoly}(P,q,r,r')$ is $O\left( (4s)^{O(d)}\log\frac{1}{\eps\rho} \right) = O(s^{O(d)})$ since $k=d$, $\lambda = 2$ and $D=2s$ in Lemma \ref{lem:sys_poly}.

%Let $\textsf{SysPoly}(P,q,r,r')$ be the following system of polynomial, which maximizes a Taylor approximation of the KDE given by the point set $Q$, subject to the constraint that the returned point has to be close to the point $q$.
%\begin{align*}
%&
%\max \sum_{p\in Q_{P,q}(r')} \prod_{i=1}^d\left( \sum_{j=0}^{s-1} \frac{1}{j!}\left(-(x_i-p_i)^2\right)^j \right) \\
%\text{s.t. } & \norm{x - q}^2 \leq r^2\log \frac{1}{\eps\rho} \\
%\text{where } & s= (r+r')^2e^2\log \frac{d}{\eps\rho}
%\end{align*}
%
%Note that we are not solving the above optimization problem directly.
%Instead, we use Lemma \ref{lem:sys_poly} to solve the following system
%\[
%	\sum_{p\in Q_{P,q}(r')} \prod_{i=1}^d\left( \sum_{j=0}^{s-1} \frac{1}{j!}\left(-(x_i-p_i)^2\right)^j \right)  \geq \beta \quad \bigwedge \quad \norm{x - q}^2 \leq r^2\log \frac{1}{\eps\rho}
%\]
%and perform binary search on $\beta$.

The following lemma captures the approximation error from the Taylor series truncation.
\begin{lemma}\label{lem:sys_poly_guar}
	Suppose $r+r'>1$ and $q\in\R^d$ such that $\norm{x^*-q} \leq r\sqrt{\log \frac{1}{\eps\rho}}$.
	Then, the output $x^{(q)}$ of $\textsf{SysPoly}(P,q,r,r')$ satisfies
	\[
	\G_{Q_{P,q}(r')}(x^{(q)}) \geq \G_{Q_{P,q}(r')}(x^*) - \abs{P}\frac{\eps \rho}{2}
	\]
\end{lemma}

We delay this technical proof to Appendix \ref{app:proofs}.  
In short, it shows that the trunctation of the above infinite summation of polynomial terms (wrapped in a sum over all points $Q$, and the product over $d$ dimensions) induces an error terms $\mathcal{E}(x^{(q)})$ and $\mathcal{E}(x^*)$ at $x^{(q)}$ and $x^*$, respectively.  We can show that the difference between these terms is at most $\eps \rho$ for our choice of $s$, as desired.

%%%%%%%%%%%%%%%%%%%%%%%%%%%%%%%%%%%%%%%%%%%%%%%%%%%%%%%%%%%%%%%%%%
%%%%%%%%%%%%%%%%%%%%%%%%%%%%%%%%%%%%%%%%%%%%%%%%%%%%%%%%%%%%%%%%%%
\subsection{Algorithm for Searching Polynomial Systems in Neighborhoods}

Now following Algorithm \ref{alg:sys_poly}, we create a set $\mathsf{G}_P$ of neighborhoods, defined by the subset of $\Gr\left(2\sqrt{\frac{\log \frac{1}{\eps\rho}}{d}}\right)$ which is within $4\sqrt{\log \frac{1}{\eps\rho}}$ of some point $p \in P$.  
%We call such $q \in \Gr(2\sqrt{\frac{\log \frac{1}{\eps\rho}}{d}})$ the set $\mathsf{G}_P$.  
And for each $q \in \mathsf{G}_P$ we define a neighborhood set $Q_{P,q}(4)$, and run the algorithm in Lemma \ref{lem:sys_poly}.  
Again we return the output with associated maximum $\G_{Q_{P,q}()}(\cdot)$ value.  

\begin{algorithm}[H]
	\caption{Solving System of Polynomial using an Infinite Grid}
	{\bf input}: a point set $P\subset \mathbb{R}^d$, parameter $\eps,\rho>0$
	%	{\bf output}: stationary point of $f(x) = \sum_{i=1}^n e^{-\norm{x - \mu_i}^2}$
	\begin{algorithmic}[1]
		\FOR {each $p\in P$}
		\STATE insert $p$ into $Q_{P,q}(4)$ for each $q\in B_p(4\sqrt{\log \frac{1}{\eps\rho}})\cap \Gr(2\sqrt{\frac{\log \frac{1}{\eps\rho}}{d}})$
		\ENDFOR
		\STATE Let $\mathsf{G}_P$ be the $q \in \Gr(2\sqrt{\frac{\log \frac{1}{\eps\rho}}{d}})$ such that $Q_{P,q}(4)$ is non empty
		\FOR {each $q \in \mathsf{G}_P$}
		\STATE Let $x^{(q)}$ be the solution to $\textsf{SysPoly}(P,q,2,4)$ by the algorithm in Lemma \ref{lem:sys_poly}
		\ENDFOR
		\STATE \textbf{return} $x' = \argmax_{q \in \mathsf{G}_P}\G_{Q_{P,q}(4)}(x^{(q)})$
	\end{algorithmic}
	\label{alg:sys_poly}
\end{algorithm}

%We also denote $x^{**}$ be $\arg\max_{x\in\mathbb{R}^d}\Gb_{P_0}(x)$ and $x'$ be the output of Algorithm \ref{alg:sys_poly} using $P_0$ as the input.

\begin{theorem}\label{thm:ig-alg}
	Given $\eps,\rho>0$ and a point set $P\subset\mathbb{R}^d$ of size $n$, we can find $x'\in\mathbb{R}^d$ so
	\[
	\Gb_P(x') \geq \Gb_P(x^*) - \eps\rho
	\]
	if $\Gb_P(x^*)\geq \rho$ where $x^*=\argmax_{x\in\mathbb{R}^d} \Gb_P(x)$ in $O\left( n\cdot\log n\cdot(2\sqrt{2e\pi})^d + n\cdot \left(\log \frac{d}{\eps\rho}\right)^{O(d)} \right)$ time. %\waiming{might need to change depending on how lemma 5 is stated} time.
\end{theorem}

%\jasper{Is the theorem considering the random subsampling coreset? There seems to be two analyses below, one not considering the coreset and one considering it. I'm not opposed to having one coreset use in this section and then have another coreset use in Section 4 that is applied before the JL projection, but maybe we can briefly discuss this.}
%\jeff{Not sure the best solution yet.  I am writing now so that two algorithms here do not use the coreset implicitly, but the theorems do.  In high-d we want to call the algorithm, but take different coresets (twice) before doing so -- so the analysis will change, but not the specific algorithm we call.  }

\begin{proof}
%	The accuracy follows mostly from the same arguments as in the proof of Theorem \ref{thm:pbn-alg}.  
	We need to argue that $x^*$ must be contained in some neighborhood $B_q(2 \sqrt{\log \frac{1}{\eps \rho}})$ for some $q \in \mathsf{G}_P$, and then apply Lemma \ref{lem:sys_poly} with $r=2$ and $r'=4$.  
	This follows since, by Lemma \ref{lem:union}, $x^*\in B_p(\sqrt{\log \frac{1}{\rho}})\subset B_p(\sqrt{\log \frac{1}{\eps\rho}})$ for some $p\in P$.  Then let $q \in \mathsf{G}_P$ be the closest grid point to that point $p$; the distance $\|p-q\| \leq \sqrt{d} \gamma/2 =\sqrt{\log \frac{1}{\eps \rho}}$ with $\gamma = 2 \sqrt{\log (1/\eps \rho) / d}$.  Then by triangle inequality $\|q - x^*\| \leq \|q-p\|+\|p-x^*\| \leq 2\sqrt{\log \frac{1}{\eps \rho}}$.  
	Now, we conclude that the output of $\textsf{SysPoly}(P,p,2,4)$ satisfies 
	\begin{align*}
	\G_P(x') 
	& \geq 
	\G_{Q_{P,q}(4)}(x') 
	= 
	\G_{Q_{P,q}(4)}(x^{(q)}) 
	\geq 
	\G_{Q_{P,q}(4)}(x^*)  - \abs{P}\frac{\eps\rho}{2} \\
	& =
	\sum_{p\in P} e^{-\norm{p - x^*}^2} - \sum_{p\notin Q_{P,q}(4)} e^{-\norm{p - x^*}^2} -\abs{P}\frac{\eps\rho}{2}
	\end{align*}
	Note that $\norm{x^*-p} \geq \norm{q-p} - \norm{q-x^*} \geq 4\sqrt{\log \frac{1}{\eps\rho}} - 2\sqrt{\log \frac{1}{\eps\rho}} = 2\sqrt{\log \frac{1}{\eps\rho}}$ since $x^*\in B_q\left(2\sqrt{\log \frac{1}{\eps\rho}}\right)$ and $p\notin B_q\left(4\sqrt{\log \frac{1}{\eps\rho}}\right)$.
	\begin{align*}
	\G_P(x') 
	& \geq
	\G_P(x^*) - \abs{Q_{P,q}(4)}(\eps\rho)^4 - \abs{P}\frac{\eps\rho}{2} \\
	& \geq 
	\G_P(x^*) - \abs{P}\eps\rho & \text{for sufficient small $\eps\rho$}
	\end{align*}
	
	We now compute the running time.
%	The runtime is a bit more complicated. 
%	We again start with a random sample $P_0$ of size $n_0 = O(\frac{d}{\eps^2\rho}(\log\frac{1}{\rho}+\log \frac{1}{\delta}))$.  
	First, to construct $Q_{P,q}(4)$ (for notation convenience, we use $Q_q$ instead), for each $p\in P$, we enumerate all $q\in B_p\left(4\sqrt{\log\frac{1}{\eps\rho}}\right) \cap \Gr\left(2\sqrt{\frac{\log\frac{1}{\eps\rho}}{d}}\right)$ and insert $p$ into $Q_q$.
	Since, by considering the volume of high dimensional sphere, there are $O\left(\frac{\pi^{d/2}}{\Gamma(\frac{d}{2}+1)}\left(4\sqrt{\log\frac{1}{\eps\rho}} \bigg/ 2\sqrt{\frac{\log\frac{1}{\eps\rho}}{d}}\right)^d\right) = O\left((2\sqrt{2e\pi})^d\right)$ points in $B_p\left(4\sqrt{\log\frac{1}{\eps\rho}}\right) \cap \Gr\left(2\sqrt{\frac{\log\frac{1}{\eps\rho}}{d}}\right)$ for each $p\in P$, we have $\sum_{q \in \mathsf{G}_P} \abs{Q_q} = O(n(2\sqrt{2e\pi})^d)$ and also there are only $O(n(4\sqrt{2e\pi})^d)$ non empty $Q_q$.
	Here, $\Gamma$ is the gamma function and we use the fact of $\Gamma(x+1) \geq (\frac{x}{e})^x$.
	It is easy to construct a data structure to insert all $p$ into all of the corresponding $Q_q$ in $O\left(n(2\sqrt{2e\pi})^d\log\left(n(2\sqrt{2e\pi})^d\right)\right) = O\left(n(2\sqrt{2e\pi})^d(\log n+d)\right)$.
	Let $s= 36e^2\log \frac{d}{\eps\rho}$.
	We now can precompute each polynomial $\prod_{i=1}^d\left( \sum_{j=0}^{s-1} \frac{1}{j!}\left(-(x_i-p_i)^2\right)^j \right)$ in $O(d(2s)^d)$ time for each $p\in P$ which takes $O(nd(2s)^d)$ total time to compute all of them.
	For each $q \in \mathsf{G}_P$, it takes $O(\abs{Q_q}(2s)^d)$ to construct the polynomial $\sum_{p\in Q_q} \prod_{i=1}^d\left( \sum_{j=0}^{s-1} \frac{1}{j!}\left(-(x_i-p_i)^2\right)^j \right)$ and $O(s^{O(d)})$ time to solve the system of polynomial as suggested in Lemma \ref{lem:sys_poly}.
	Therefore, the total running time is 
	\begin{align*}
	\MoveEqLeft	O\left( n (2\sqrt{2e\pi})^d(\log n+d) + n d(2s)^d + \sum_{q \in \mathsf{G}_P}(\abs{Q_q}(2s)^d + s^{O(d)} \right) 
	\\&= 
	O\left( n\cdot\log n\cdot(2\sqrt{2e\pi})^d + n\cdot \left(\log \frac{d}{\eps\rho}\right)^{O(d)} \right) \qedhere
	\end{align*}
\end{proof}

%%%%%%%%%%%%%%%%%%%%%%%%%%%%%%%%%%%%%%%%%%%%%%%%%%%%%%%%%%%%%%%%%%
%%%%%%%%%%%%%%%%%%%%%%%%%%%%%%%%%%%%%%%%%%%%%%%%%%%%%%%%%%%%%%%%%%
%%%%%%%%%%%%%%%%%%%%%%%%%%%%%%%%%%%%%%%%%%%%%%%%%%%%%%%%%%%%%%%%%%
\section{Dimensionality reduction for KDE mode finding}
\label{sec:JL+KDS}
In the previous section we described an algorithm for finding the mode, or maximizing point, of a Gaussian kernel density estimate when the data dimension $d$ is low. The approach becomes intractable in high dimensions because it has an exponential dependence on $d$.
In this section, we show how to avoid that dependence by providing, to the best of our knowledge, the first dimensionality reduction result for KDE mode finding. 

In particular, by leveraging Kirszbraun's extension theorem, we prove that compressing $P = \{{p}_1,\ldots,p_n\}$ using a Johnson-Lindenstrauss random projection to $O\left(\log n \log^2(1/\eps\rho)/\eps^2\right)$ dimensions preserves the mode of the KDE with centers in $P$, to a $(1 - \eps)$ factor.
%\footnote{By using existing coreset results for KDEs, the $\log n$ dependence can be improve to $\log(1/\epsilon\rho)$.}. 
Crucially, we then 
show that it is possible to {recover} an approximate mode for $P$ from a solution to the low dimensional problem by applying a \emph{single iteration} of the mean-shift algorithm. 

In Section \ref{sec:d-large} we combine this result with our low dimensional algorithm from Section \ref{sec:d-small} and existing coreset results for KDEs (which allow us to eliminate the $\log n$ dependence) to give our final algorithm for high-dimensional mode finding. We first present the dimensionality reduction result in isolation as, like dimensionality reduction strategies for other computational hard problems \cite{CohenElderMusco:2015,Makarychev:2019,Becchetti:2019}, it could  in principal
be combined with any other heuristic or approximate mode finding method. For example, we expect a practical strategy would be to solve the low-dimensional problem using the mean-shift heuristic.

We need one basic definition before outlining our approach in Algorithm \ref{alg:dimensionality_reduction}.

\begin{definition}[$(\gamma,k,\delta)$-Johnson-Lindenstrauss Guarantee]
	\label{def:jl}
	A randomly selected matrix $\Pi \in \R^{m\times d}$ satisfies the $(\gamma,k,\delta)$-JL Guarantee if, for any $k$ data points $v_1, \ldots, v_k \in \R^d$,
	\begin{align*}
	\|v_i - v_j\| \leq \|\Pi v_i - \Pi v_j\| \leq ( 1+ \gamma)\|v_i - v_j\| ,
	\end{align*}
	for all pairs $i,j \in 1,\ldots, k$ simultaneously, with probability $(1-\delta)$.
\end{definition}
Definition \ref{def:jl} is satisfied by many possible constructions. When $\Pi$ is a properly scaled random Gaussian or sign matrix, it satisfies the $(\gamma,k,\delta)$-JL guarantee as long as $m = O(\log(k/\delta)/\gamma^2)$ \cite{DasguptaGupta:2003,Achlioptas:2003}. In this case, $\Pi$ can be multiplied by a $d$ dimensional vector in $O(md)$ time. For simplicity, we assume such a construction is used in our algorithm. Other constructions, including fast Johnson-Lindenstrauss transforms \cite{AilonChazelle:2009,AilonLiberty:2013,KrahmerWard:2011} and sparse random projections \cite{KaneNelson:2014,CohenJayramNelson:2018} satisfy the definition with slightly larger $m$, but faster multiplication time. Depending on problem parameters, using such constructions may lead to a slightly faster overall runtime.

\begin{algorithm}[H]
	\caption{Dimensionality Reduction for KDE mode finding}
	{\bf input}: a set of $n$ points $P\subset \R^d$, parameters $\eps, \delta > 0$, $\rho$ such that $\max_x G_P(x) \geq \rho n$\\
	{\bf output}: a point $x' \in \R^d$ satisfying  $G_P(x')\geq(1-\epsilon)\max_x G_P(x)$ with prob. $1-\delta$
	
	%	{\bf output}: stationary point of $f(x) = \sum_{i=1}^n e^{-\norm{x - \mu_i}^2}$
	\begin{algorithmic}[1]
		\STATE Set $\gamma = \frac{\epsilon}{4\log(4/\epsilon \rho)}$.
		\STATE Choose a random matrix $\Pi \in \R^{m\times d}$ satisfying the $(\gamma,n+1,\delta)$-JL guarantee (Defn. \ref{def:jl}).
		\STATE For each $p_i \in P$, compute $\Pi p_i$ and let $\Pi P$ denote the data set $\{\Pi p_1, \ldots, \Pi p_n\}$
		\STATE Using an algorithm for mode finding in low dimensions (e.g. Algorithm~\ref{alg:sys_poly}) find a point $x''$ satisfying $\G_{\Pi P}(x'') \geq (1-\eps/2)\max_{x \in \R^m} \G_{\Pi P}(x)$.
		\STATE \textbf{return} $x' = \frac{\sum_{p\in P} p \cdot e^{-\norm{x'' - \Pi p}^2} }{\sum_{p\in P} e^{-\norm{x'' - \Pi p}^2}}$ \label{line:mean_shift}
	\end{algorithmic}
	\label{alg:dimensionality_reduction}
\end{algorithm}

\begin{theorem}\label{thm:dim_reduction}
	With probability $(1-\delta)$, Algorithm \ref{alg:dimensionality_reduction} returns an $x'$ satisfying $\G_P(x') \geq (1-\epsilon)\max_x \G_P(x)$. When implemented with a random Rademacher or Gaussian $\Pi$, the algorithm runs in time
	$O\left(ndm\right) + T_{m,(1-\epsilon/2)}$,
	where $m = O\left(\frac{\log(n/\delta)\log^2(1/\epsilon\rho)}{\epsilon^2}\right)$ and $T_{m,(1-\epsilon)}$ is the time required to compute a $(1-\epsilon/2)$ approximate mode for an $O\left(m\right)$ dimension dataset.
\end{theorem}
The runtime claim is immediate, so we focus on proving the correctness of Algorithm \ref{thm:dim_reduction}.
We begin with the following key lemma, which claims that the mode of our dimensionality reduced problem has approximately the same density  as that of the original problem.
\begin{lemma}\label{lem:value_approx}
	\begin{align}
	\label{eq:value_approx}
	(1-\epsilon/2)\max_{x \in \R^d}\G_P(x) \leq \max_{x\in\mathbb{R}^m} \G_{\Pi P}(x) \leq \max_{x \in \R^d}\G_P(x)
	\end{align}
\end{lemma}
\begin{proof}
	Let $x^* = \argmax_{x}\G_P(x)$. Since $\Pi$ was chosen to satisfy the $(\gamma,n+1,\delta)$ property with $\gamma = \frac{\epsilon}{4\log(4/\epsilon \rho)}$, we have that, with probability at least $1-\delta$, for all $y, z \in \{x^*\}\cup P$, 
	\begin{align}
	\label{eq:jl_guar}
	\norm{y - z}^2 \leq \norm{\Pi y - \Pi z}^2 \leq \left(1+\frac{\eps}{4\log(4/\eps\rho)}\right)\norm{y -  z}^2.
	\end{align}
	The rest of our analysis conditions on this fact being true. We first prove the left side of \eqref{eq:value_approx}. 
From \eqref{eq:jl_guar}, we have that $\norm{\Pi x^* - \Pi p}^2 \leq (1+\frac{\eps}{4\log({4}/{\eps\rho})}) \norm{x^* - p}^2$ for all $p\in P$. 
Accordingly,	
\begin{align}\label{eq:sum_breakdown}
\max_{x\in\mathbb{R}^m} \G_{\Pi P}(x)  \geq \G_{\Pi P}(\Pi x^*) 
= 
\sum_{p\in P} e^{-\norm{\Pi x^* - \Pi p}^2} 
& \geq 
\sum_{p\in P} e^{-(1+\frac{\eps}{4\log(4/\eps\rho)})\norm{ x^* - p}^2} \nonumber \\
&\geq \hspace{-1.5em}
\sum_{\substack{p\in P\\ \norm{ x^*- p}^2 < \log(4/\eps\rho)}} \hspace{-1.5em} e^{-\norm{ x^* - p}^2}e^{-\frac{\eps}{4\log(4/\eps\rho)}\norm{ x^*- p}^2} \nonumber \\
& \geq 
(1-\eps/4) \hspace{-1.5em}\sum_{\substack{p\in P\\ \norm{ x^*- p}^2 < \log ({4}/{\eps\rho})}} \hspace{-1.5em} e^{-\norm{ x^* - p}^2}.
\end{align}
The last step uses that $e^{-\frac{\eps}{4\log (4/{\eps\rho})}\norm{ x - p}^2} \geq 1-\eps/4$ when $\|x-p\|^2 \leq \log({4/\eps\rho})$. 
Next we have
\begin{align*}
\sum_{\substack{p\in P\\ \norm{ x^*- p}^2 < \log \frac{4}{\eps\rho}}}  \hspace{-2em} e^{-\norm{ x^* - p}^2} 
\geq 
\sum_{p\in P} e^{-\norm{ x^* - p}^2}  - \eps n\rho
= 
\G_P(x^*) - \eps n\rho
\geq 
(1-\eps/4)\G_P(x^*).
\end{align*}
This statement follows from two facts: 1) If $\|x-p\|^2 \geq \log\frac{4}{\eps\rho}$ then $e^{-\norm{ x - p}^2} \leq  \eps\rho/4$ and 2) we assume that $\G_P(x^*) \geq \rho n$. 
Combining with \eqref{eq:sum_breakdown} we conclude that $
\G_{\Pi P}(x) \geq (1-\eps/2)\G_P(x^*)
$.

We are left to prove the right hand side of \eqref{eq:value_approx}. To do so, we rely on the classic Kirszbraun extension theorem for Lipschitz functions, which is stated as follows:
\begin{theorem}[Kirszbraun Theorem \cite{Kirszbraun:1934,Valentine:1945}]\label{thm:kirszbraun}
	For any $\mathcal{S} \subset \R^z$, let $f: S \rightarrow \R^w$ be an $L$-Lipschitz function: for all $x,y \in \mathcal{S}$, $\|f(x)-f(y)\|_2 \leq L \|x-y\|_2$. Then there always exists some extension $\tilde{f}: \R^z \rightarrow \R^w$ of $f$ to the entirety of $\R^z$ such that:
	\begin{enumerate}
		\item $\tilde{f}(x) = f(x)$ for all $x \in S$,
		\item $\tilde{f}$ is also $L$-Lipschitz: for all $x,y \in \R^z$, $\|\tilde{f}(x)-\tilde{f}(y)\|_2 \leq L \|x-y\|_2$.
	\end{enumerate}
\end{theorem}

	We will apply this theorem to the function $g: \{\Pi x^*\}\cup \Pi P \rightarrow \{x^*\}\cup P$ with $g(\Pi y) = y$ for any $y\in \{x^*\}\cup P$. By \eqref{eq:jl_guar}, we have that $g$ is $1$-Lipschitz. It follows that there is some function $\tilde{g}: \R^m \rightarrow \R^d$ which agrees with $g$ on inputs $\{\Pi x^*\} \cup P$ and satisfies
	$\norm{\tilde{g}(s) - \tilde{g}(t)} \leq \norm{s - t}$
	for all $s,t \in \R^m$.
	This fact can be used to establish  that, for any $x \in \R^m$, $\G_{\Pi P}(x) \leq \G_P(\tilde{g}(x))$:
	\begin{align*}
		\G_{\Pi P}(x) 
		= 
		\sum_{p\in P} e^{-\norm{x - \Pi p}^2}  
		\leq 
		\sum_{p\in P} e^{-\norm{\tilde{g}(x) - \tilde{g}(\Pi p)}^2}  
		= 
		\sum_{p\in P} e^{-\norm{\tilde{g}(x) - p}^2} = \G_P(\tilde{g}(x)).   
	\end{align*}
	It thus follows that $\max_{x} \G_{\Pi P}(x)\leq  \max_x\G_P(x)$, so the right side of \eqref{eq:value_approx} is proven. 
\end{proof}

Note that in proving Lemma \ref{lem:value_approx} we have also proven the following statement.
\begin{corollary}\label{cor:dist_approx}
	For any $x \in \R^m$, there exists some point $\tilde{g}(x) \in \R^d$ such that, for all $p\in P$,
	$\|\tilde{g}(x) - p\| \leq \|x - \Pi p\|$.
\end{corollary}

Next, we complete the proof of Theorem \ref{thm:dim_reduction} by showing that, not only does the maximum of $\G_{\Pi P}$ approximate that of $\G_{P}$, but an approximate maximizer for $\G_{\Pi P}$ can be used to recover one for $\G_{P}$.
%With  showed that random projection to $m = O\left(\frac{\log \frac{n}{\delta} \log^2\frac{1}{\eps\rho}}{\eps^2}\right)$ is sufficient to approximate the value of $\G_P(x^*)$ to within a $(1- \eps)$ multiplicative factor. In this section we show how to actually recover a solution $x'$ in the original $d$ dimensional space  with $\G_P(x') \geq (1-\eps)\G_P(x^*)$.
Algorithm \ref{alg:dimensionality_reduction} does so on Line \ref{line:mean_shift} by applying what can be viewed as single iteration of the \emph{mean-shift} algorithm, a common heuristic KDE mode finding \cite{Carreira-Perpinan:2000, Carreira-Perpinan:2007}:

\begin{algorithm}[H]
	\caption{Mean-shift Algorithm}
	{\bf input}: a set of $n$ points $P\subset \R^d$, number of iterations $t$.
	
%	{\bf output}: a point $x^{(t)} \in \R^d$ which is hopefully a near maximizer of $G_P(x)$
	%	{\bf output}: stationary point of $f(x) = \sum_{i=1}^n e^{-\norm{x - \mu_i}^2}$
	\begin{algorithmic}[1]
		\STATE Select some initial point $x^{(0)} \in \R^d$
		\FOR{$i = 0, \ldots, t-1$}
		\STATE $x^{(i+1)} = \frac{\sum_{p\in P} p \cdot e^{-\norm{x^{(i)}-p}^2}}{\sum_{p\in P} e^{-\norm{x^{(i)}-p}^2}}$ \label{line:update}
		\ENDFOR
		\STATE \textbf{return} $x^{(t)}$
	\end{algorithmic}
	\label{alg:mean_shift}
\end{algorithm}
While not guaranteed to converge to a point which maximizes $\G_P$, a useful property of the mean-shift algorithm is that its solution is guaranteed to improve on each iteration:
\begin{claim}
	\label{claim:monotonic_mean_shift}
	Given $y \in \R^d$, let $y' = \frac{\sum_{p\in P} p \cdot e^{-\norm{y-p}^2}}{\sum_{p\in P} e^{-\norm{y-p}^2}}$, then $\G_P(y') \geq \G_P(y)$. 
\end{claim}
\begin{proof} We prove this well known fact for completeness. First, observe by rearrangement that $\G_P(y') - \G_P(y) = \sum_{p\in P} \left(e^{-\norm{y' - p}^2+\norm{y-p}^2} - 1\right)e^{-\norm{y-p}^2}$. Then, since $e^z \geq 1+z$ for all $z$,
	\begin{align*}
		\G_P(y') - \G_P(y) &\geq
		\sum_{p\in P} \left(1-\norm{y' - p}^2+\norm{y-p}^2 - 1\right)e^{-\norm{y-p}^2} \\
		&= 
		\sum_{p\in P} -\left(\norm{y'}^2 - \norm{y}^2 - 2(y' - y)^Tp\right)e^{-\norm{y - p}^2}\\
		&= 
		-\left(\norm{y'}^2 - \norm{y}^2\right)\sum_{p\in P} e^{-\norm{y - p}^2} + 2(y' - y)^T\sum_{p\in P}pe^{-\norm{y - p}^2}\\
		&= 
		\G_P(y)\left(-\norm{y'}^2 + \norm{y}^2 + 2(y' - y)^Ty'\right)\\
		&= 
		\G_P(y)\norm{y' - y}^2
		\geq 
		0.  \qedhere
	\end{align*}
\end{proof}
With Claim \ref{claim:monotonic_mean_shift} in place, we can prove the main result of this section:

\begin{proof}[Proof of Theorem \ref{thm:dim_reduction}]
	Recall from Corollary \ref{cor:dist_approx} that for any $x$, there is always a $\tilde{g}(x)$ with 
	\begin{align}
	\label{eq:inequal_again}
	\|\tilde{g}(x) - p\| \leq \|x - \Pi p\|
	\end{align}
	for all $p\in P$. 
	Suppose this inequality was tight: i.e., suppose that for all $p\in P, x\in \R^m$, $\|\tilde{g}(x) - p\| = \|x - \Pi p\|$. Then letting $x''$ be as defined in Algorithm \ref{alg:dimensionality_reduction}, we would have that Line \ref{line:mean_shift} sets $x'$ equal to a mean-shift update applied to $\tilde{g}(x)$.
	From Claim \ref{claim:monotonic_mean_shift} we would then immediately have that $\G_P(x') \geq \G_P(\tilde{g}(x'')) = \G_{\Pi P}(x'') \geq (1-\eps/2)\max_x\G_{\Pi P}(x) \geq (1-\eps)\max_x\G_{P}(x)$, which would prove the theorem.
	
	However, since \eqref{eq:inequal_again} is not tight, we need a more involved argument.
	In particular, for each $p\in P$, let $\bar{p} \in \R^{d+1}$ be a vector with its first $d$ entries equal to $p$ and let the final entry be equal to 
	\begin{align*}
	\sqrt{\|x - \Pi p\|^2 - \|\tilde{g}(x) - p\|^2}.  
	\end{align*}
	Additionally, let $\bar{\tilde{g}}(x) \in \R^{d+1}$ be a vector with its first $d$ entries equal to $\tilde{g}(x)$ and final entry equal to $0$. Clearly, for any $p\in P$,
	\begin{align}
	\label{eq:bar_again}
	\|\bar{\tilde{g}}(x) - \bar{p}\| =  \|x - \Pi p\|.
	\end{align}
	For $z\in \R^{d+1}$, let $\bar{\G}_P(z) =  \sum_{p\in P} e^{-\norm{z - \bar{p}}^2}$ and let $\bar{x}' = \frac{\sum_{p\in P} \bar{p} e^{-\norm{x'' - \Pi p}^2} }{\sum_{p\in P} e^{-\norm{x'' - \Pi p}^2}}$.
	It follows from \eqref{eq:bar_again} and  the argument above that $\bar{\G}_P(\bar{x}') \geq \bar{\G}_P(\bar{\tilde{g}}(x')) = \G_{\Pi P}(x'')$. But clearly it also holds that $\G_P(x')\geq \bar{\G}_P(\bar{x}')$ because, for any $p\in P$, $\|x' - p\| \leq \|\bar{x}' - \bar{p}\|$. So we conclude that $\G_P(x')\geq \G_{\Pi P}(x'')$ as desired.
	Furthermore, recall that $x''$ is an approximate mode in the projected setting. It satisfies $\G_{\Pi P}(x'') \geq \max_x(1-\eps/2)\G_{\Pi P}(x)$, and from Lemma~\ref{lem:value_approx} we have that $\max_x\G_{\Pi P}(x) \ge (1-\eps/2)\max_x\G_P(x)$.
	Chaining these inequalities gives the desired bound that $\G_P(x') \ge (1-\eps/2)^2\max_x\G_P(x) \ge (1-\eps)\max_x\G_P(x)$.
	 %and, by our assumption of $\G_{\Pi P}(x'') \geq (1-\eps)\G_{\Pi P}(x^{**})$ and Lemma \ref{lem:value_approx1}, therefore $\G_P(x')\geq \G_{\Pi P}(x'') \geq (1-\eps)\G_{\Pi P}(x^{**}) \geq (1-\eps)^2\G_P(x^*) = (1-O(\eps))\G_P(x^*)$.
\end{proof}

\section{Final Algorithms for $d \ge 3$}
\label{sec:finalschemes}

In this section, we combine the previous polynomial solving and dimensionality results with standard coreset results to give our final algorithms.
Our first algorithm has runtime dependent on the dimensionality, and is suitable for low dimensions.
Our second algorithm on the other hand removes all dependence on the dimensionality beyond the linear dependence for reading the input and basic pre-processing, which is useful for instances where the ambient dimension is high.

The next section covers the special case where $d=2$, where we instead give a specialized combinatorial algorithm.

%%%%%%%%%%%%%%%%%%%%%%%%%%%%%%%%%%%%%%%%%%%%%%%%%%%%%%%%%%%%%%%%%%
%%%%%%%%%%%%%%%%%%%%%%%%%%%%%%%%%%%%%%%%%%%%%%%%%%%%%%%%%%%%%%%%%%
\subsection{Coreset Results We Leverage}
\label{sec:coresets}

There are now a variety of results for coresets for kernel density estimates $\Gb_P$; see \cite{phillips2018near} for an overview.  These are a subset $Q \subset P$ so $\Gb_Q$ approximates $\Gb_P$.  They are typically stated in terms of additive error, so $\max_{x \in \R^d} \abs{\Gb_P(x) - \Gb_Q(x)} \leq \alpha$ for a point set $P$ of size $n$.    In low dimensions $d$, Phillips and Tai~\cite{phillips2018near} provided a coreset of size $O((\sqrt{d}/\alpha) \sqrt{\log 1/\alpha})$ that runs in time $O(n \cdot \mathrm{poly}(1/\alpha))$.  Recently Karnin and Liberty~\cite{KL19} showed coresets of size $O(\sqrt{d}/\alpha)$ exist (with no algorithm), which is tight for constant $d$~\cite{phillips2018near}.  
In high dimensions, Lopaz-Paz \etal~\cite{LopazPaz15} showed a random sample of size $O(1/\alpha^2 \log \frac{1}{\delta})$ is an additive error coreset with probability $1-\delta$, and this again tight~\cite{phillips2018near}.  
By setting $\alpha = \eps \rho$, each of these results can provide a relative $(1-\eps)$-approximation guarantee in $\Gb_{P_0}(x)$ for all $x \in \R^d$ such that $\Gb_P(x) > \rho$.
Alternatively Zheng and Phillips~\cite{zheng2017coresets} directly showed that a random sample $P_1$ of size $O(\frac{d}{\eps^2}\frac{1}{\rho} (\log \frac{1}{\rho} + \log \frac{1}{\delta}))$ could guarantee a $(\eps,\rho)$-approximation where for any $x \in \R^d$ that
$\abs{\Gb_P(x) - \Gb_{P_1}(x)} \leq \eps M_x$ where $M_x = \max\{\Gb_P(x),\rho\}$.  

In the low-dimensional case, we will start with the coreset of Zheng and Phillips~\cite{zheng2017coresets} to create a coreset $P_1$ of size $n_1 = O(\frac{d}{\eps^2}\frac{1}{\rho} (\log \frac{1}{\rho} + \log \frac{1}{\delta}))$, which makes the dependence on $1/\rho$ small.  
The random sample can be computed in the time $O(nd)$ it takes to read the data.  
%the runtime is simply $O(n_1 d) = O(nd)$.

In the high-dimensional case, we design a constant-probability algorithm that we repeat $O(\log\frac{1}{\delta})$ times.
We will first use the sampling result of Lopaz-Paz \etal~\cite{LopazPaz15} to create coresets $P^j_0$ (for $j \in [1..O(\log\frac{1}{\delta})]$) with size $n_0 = O(\frac{1}{\eps^2 \rho^2})$ independent of $d$.
All $O(\log\frac{1}{\delta})$ samples can be constructed with a single pass of the data, with runtime $O(nd + n_0 \log\frac{1}{\delta})$.
Then, after dimensionality reduction to dimension $m$, we apply the coreset result of Phillips and Tai~\cite{phillips2018near} to obtain another coreset $P^j_2$ of size $n_2 = O((\sqrt{m}/\eps\rho) \sqrt{\log (1/\eps \rho)})$ which reduces the dependence on $1/\rho$.
The runtime of this coreset construction is $\poly(n_0,m,1/\eps\rho)$, and as we shall see, will be dominated by the quasipolynomial time needed for polynomial system solving.
%\jasper{Does it have a nice-looking runtime?}
%\jeff{We have not worked out the runtime.  Its not too high (something like degree 4 or 5), but is messy.}
%\jasper{Also, if we have space we should write down the coreset results cleanly as propositions}
%\jeff{Looks like we are going to be tight on space.  Also, I think (my two cents) these sort of things can be superfluous and obscure what are the main things contributed in the paper.  It is fine in the case of Renegar and Kirszbraun since the actual statements are technical and delicate.  But for these cases I think it is easier to state in text.  }
% Fair points :)

%%%%%%%%%%%%%%%%%%%%%%%%%%%%%%%%%%%%%%%%%%%%%%%%%%%%%%%%%%%%%%%%%%
%%%%%%%%%%%%%%%%%%%%%%%%%%%%%%%%%%%%%%%%%%%%%%%%%%%%%%%%%%%%%%%%%%
\subsection{Algorithm for Low Dimension}
We pre-process the input $P$ by constructing, under the assumption that $\max_x \Gb_P(x) \ge \rho$, a $(1-\eps/3)$-approximation coreset from \cite{zheng2017coresets} of size $n_1 = O(\frac{d}{\eps^2}\frac{1}{\rho} (\log \frac{1}{\rho} + \log \frac{1}{\delta}))$.
%This coreset can achieve the following guarantee with probability at least $1-\delta$:
%For any $x\in \mathbb{R}^d$, we have $\abs{\Gb_P(x) - \Gb_{P_1}(x)} \leq \frac{1}{3}\eps M_x$ where $M_x=\max\{\Gb_P(x), \rho\}$.
Then we run Algorithm \ref{alg:sys_poly} on $P_1$ to get solution $x'$.

\begin{theorem}[restated Theorem \ref{thm:low}]\label{thm:main_low}
	Given $\eps,\rho>0$ and a point set $P\subset\mathbb{R}^d$ of size $n$, we can find $x'\in\mathbb{R}^d$ so
	\[
	\Gb_P(x') \geq (1-\eps)\Gb_P(x^*)
	\]
	if $\Gb_P(x^*)\geq \rho$ where $x^*=\argmax_{x\in\mathbb{R}^d} \Gb_P(x)$ in $O\left(nd+\frac{d}{\eps^2\rho} \cdot \log\frac{1}{\rho\delta} \cdot \left(\log \frac{d}{\eps\rho}\right)^{O(d)}\right)$ time with probability $1-\delta$.
\end{theorem}

\begin{proof}
	Let $x^{**} = \argmax_{x\in\R^d} \Gb_{P_1}(x)$.
	We first have $\Gb_{P_1}(x^{**}) \geq \Gb_{P_1}(x^*) \geq (1-\frac{1}{3}\eps)\Gb_P(x^*) = \Omega(\rho)$ for small $\eps$.
	By Lemma \ref{thm:ig-alg} and reparameterizing $\eps$, we have 
	\begin{align*}
	\Gb_P(x')
	& \geq 
	\Gb_{P_1}(x') -\frac{1}{3}\eps M_{x'} & \text{by the construction of $P_1$}\\
	& \geq
	\Gb_{P_1}(x^{**})-\frac{1}{3}\eps\rho -\frac{1}{3}\eps M_{x'} & \text{since $\Gb_{P_1}(x^{**}) = \Omega(\rho)$ and by Lemma \ref{thm:ig-alg}}\\
	& \geq
	\Gb_{P_0}(x^{*}) - \frac{1}{3}\eps\rho -\frac{1}{3}\eps M_{x'} \\
	& \geq
	(1-\frac{1}{3}\eps) \Gb_P(x^*) -\frac{1}{3}\eps\rho - \frac{1}{3}\eps M_{x'} & \text{by $\Gb_P(x^*) \geq \rho$ and construction of $P_1$}\\
	& \geq
	(1-\eps)\Gb_P(x^*) & \text{since $M_{x'}\leq \Gb_P(x^*)$} 
	\end{align*}
	The final running time is $O(nd)$ to read data and construct $P_1$ plus
	\[
	O\left( n_1\cdot\log n_1\cdot(2\sqrt{2e\pi})^d + n_1\cdot \left(\log \frac{d}{\eps\rho}\right)^{O(d)} \right) = O\left(\frac{d}{\eps^2\rho} \cdot \log\frac{1}{\rho\delta} \cdot \left(\log \frac{d}{\eps\rho}\right)^{O(d)}\right). \qedhere
	\]	
\end{proof}

\subsection{Algorithm for High Dimension}
\label{sec:d-large}

For high dimensional case, we combine together the techniques of 1) dimensionality reduction, 2) polynomial system solving and 3) coresets that we have developed to obtain an algorithm that is linear in the dimensionality $d$ and exponential only in $\poly(1/\eps, \log 1/\rho)$.
We use Algorithm~\ref{alg:dimensionality_reduction} as a subroutine, where we instantiate Step 4 (solving mode finding in low dimensions) with an application of a coreset result combined with the polynomial system solving approach from Section~\ref{sec:d-small}.
%We first pre-process the data using dimensionality reduction, then use a system of polynomials as in the case for bounded $d$.  
%Once we obtain an approximate solution in low dimensional space, we can actually recover the solution in the original space.
We write out the full algorithm as follows for easier reading, where Steps 3-8 correspond to Algorithm~\ref{alg:dimensionality_reduction}.

%Before we describe the full algorithm, the following lemma is useful in our analysis.
%\begin{lemma}[\cite{phillips2018near}]\label{lem:coreset}\jasper{Move to 4.1}
%	Given a point set $P\subset \R^d$ and $\eps,\delta>0$, there is an algorithm to construct a subset $Q\subset P$ of size $O(\frac{\sqrt{d}}{\eps}\sqrt{\log \frac{1}{\eps\delta}})$ such that $\abs{\Gb_P(x) - \Gb_Q(x)} \leq \eps$ for all $x\in\R^d$ in $O(\abs{P}\poly(\frac{1}{\eps}))$ time with probability at least $1-\delta$.
%\end{lemma}

%We can now describe the full algorithm.
\begin{algorithm}[H]
	\caption{Full algorithm for high dimensional case}
	{\bf input}: a point set $P\in \R^d$, parameter $\eps,\rho,\delta>0$
	%	{\bf output}: stationary point of $f(x) = \sum_{i=1}^n e^{-\norm{x - \mu_i}^2}$
	\begin{algorithmic}[1]
		\STATE Generate $O(\log\frac{1}{\delta})$ random samples $P^j_0\subset P$ of size $n_0 = O(\frac{1}{\eps^2\rho^2})$ (\`{a} la Lopaz-Paz et al.)
		\FOR {$j \gets 1$ to $O(\log\frac{1}{\delta})$}
		
		\STATE Set $\gamma = \frac{\epsilon}{4\log(4/\epsilon \rho)}$.
		\STATE Choose random matrix $\Pi \in \R^{m\times d}$ satisfying $(\gamma,n+1,1/100)$-JL guarantee (Defn. \ref{def:jl})
		\STATE For each $p_i \in P^j_0$, compute $\Pi p_i$ and let $\Pi P^j_0$ denote the data set $\{\Pi p_1, \ldots, \Pi p_n\}$
		%\STATE Using an algorithm for mode finding in low dimensions (e.g. Algorithm~\ref{alg:sys_poly}) find a point $x''$ satisfying $\G_{\Pi P}(x'') \geq (1-\eps/2)\max_{x \in \R^m} \G_{\Pi P}(x)$.
		
		%\STATE construct $\Pi P_0$ by a random JL matrix $\Pi$ with $m = O(\frac{1}{\eps^2}\log \frac{n_0}{\delta}\log^2\frac{1}{\eps\rho}) = O(\frac{1}{\eps^2}\log\frac{1}{\eps\rho\delta}\log^2\frac{1}{\eps\rho})$
		
		\STATE Run the algorithm in Phillips and Tai~\cite{phillips2018near} to construct a subset $P^j_2\subset \Pi P^j_0$ of size $n_2 = O(\frac{\sqrt{m}}{\eps\rho}\sqrt{\log \frac{1}{\eps\rho}}) = O(\frac{1}{\eps^2\rho}\log^2\frac{1}{\eps\rho})$ 

          \STATE Set $x''$ as the output of Algorithm \ref{alg:sys_poly} (Section~\ref{sec:d-small}) on $P^j_2$ in dimension $m$  \label{line:pbn-alg}
%		\For {each $q\in P_1$}\label{alg:loop}
%		\State construct $Q_q = P_1 \cap B_q(2\sqrt{\log \frac{1}{\eps\rho}})$ by linear search
%		\State solve $\textsf{SysPoly}(Q_q, q, 1,m)$ by the algorithm in Lemma \ref{lem:sys_poly}
%		\State let the output of the above system be $x^{(q)}$
%		\EndFor\label{alg:endloop}
%		\State let $x'' = \arg\max_{x^{(q)}}\G_{Q_q}(x^{(q)})$

		\STATE Compute new $x' = \frac{\sum_{p\in P^j_0} p \cdot e^{-\norm{x'' - \Pi p}^2} }{\sum_{p\in P^j_0} e^{-\norm{x'' - \Pi p}^2}}$ \label{line:mean_shift-full}
		\ENDFOR
		\STATE \textbf{Return} the best solution from all iterations of Step 8, evaluated on $\bigcup_j P^j_0$
	\end{algorithmic}
	\label{alg:full}
\end{algorithm}

\begin{theorem}[restated Theorem \ref{thm:high}]
	Given $\eps,\rho>0$ and a point set $P\subset\mathbb{R}^d$ of size $n$, we can find $x'\in\mathbb{R}^d$ so 
	\[
	\Gb_P(x') \geq (1-\eps)\Gb_P(x^*)
	\]
	if $\Gb_P(x^*)\geq \rho$, with probability at least $1-\delta$, where $x^*=\argmax_{x\in\mathbb{R}^d} \Gb_P(x)$ in time
	$O\left(nd +  
                \left(\log \frac{1}{\eps \rho } \right)^{O(\frac{1}{\eps^2}\log^3\frac{1}{\eps\rho})} \cdot \log\frac{1}{\delta} + 
                \min\{nd \log \frac{1}{\delta}, \frac{d}{\eps^2\rho^2}\log^2\frac{1}{\delta}\}\right)$.
\end{theorem}

\begin{proof}
	We first show the approximation guarantee.
	It suffices to prove that an iteration of the for loop succeeds with constant probability, so we fix a particular $j$ and omit the superscript in $P_0$ and $P_2$.
	From Lemma \ref{lem:union}, $x^*\in B_q\left(\sqrt{\log \frac{1}{\rho}}\right) \subset B_q\left(\sqrt{\log \frac{1}{\eps\rho}}\right)$ for some $q\in P_2$.
	Let $x_0^{**}$ be $\arg\max_{x\in\R^m}\G_{ {P_2}}(x)$.
	By Lemma \ref{lem:sys_poly_guar} with $r=1$, we have $\Gb_{{P_2}}(x'') \geq \Gb_{ {P_2}}(x_0^{**}) - \eps \rho \geq(1-\eps)\Gb_{{P_2}}(x_0^{**})$.
	Phillips and Tai's coreset result~\cite{phillips2018near} implies,  both $\abs{\Gb_{\Pi P_0}(x'') - \Gb_{P_2}(x'')}\leq \eps\rho$ and $\abs{\Gb_{\Pi P_0}(x^{**}) - \Gb_{P_2}(x^{**})}\leq \eps\rho$ which implies $\Gb_{\Pi P_0}(x'') \geq (1-O(\eps)) \Gb_{\Pi P_0} (x^{**})$.
	Now, let $x_0^*$ be $\arg\max_{x\in\R^d}\G_{ \Pi P_0}(x)$.
	By Theorem \ref{thm:dim_reduction}, with constant probability we have $\G_{P_0}(x') \geq (1-O(\eps))\G_{P_0}(x_0^*)$.
	By Lopaz-Paz \etal~\cite{LopazPaz15}, a random sample $P_0\subset P$ of size $n_0=O(\frac{1}{\eps^2\rho^2})$ is sufficient to have the guarantee of $\abs{\Gb_P(x) - \Gb_{P_0}(x)}\leq \eps\rho$ for any $x\in \R^d$.  
	If we combine this inequality and the guarantee of random sampling, we can conclude that $\Gb_P(x') \geq (1-\eps)\Gb_P(x^*)$.

	We now analyze the running time.
	Reading the input and constructing the coresets $P^j_0$ take 
	$O(nd + n_0\log\frac{1}{\delta})$ time in total.
	Evaluating all the solutions in Step 9 takes $O(n_0 d \log^2\frac{1}{\delta})$ time, since there are $O(\log\frac{1}{\delta})$ many candidates evaluated over a coreset of size $O(n_0 \log \frac{1}{\delta})$ in $d$ dimensions.
	%$O(nd + n_0  \log \frac{1}{\delta})$ time in total. 
%	\jasper{this is accurate right?} 
%	\jeff{Ah there is a simple solution for $O(nd)$ or $O(nd + n_0)$ time, and with assumed $n_0 < n$, then is $O(nd)$ time.  Read the data and store in an array.  Generate $n_0$ random indexes in $[n]$; these are the samples.  Under RAM this is $O(nd + n_0)$.  There is perhaps a slight technical issue that perhaps this requires a floor operation, or needs $O(\log n)$ bits, or if randomness is really allowed under RAM.  But I'd rather just avoid these issues and say $O(nd)$.}
%\jeff{You might be correct that the $\log \frac{1}{\delta}$ term is needed.  I think that is possible.  Otherwise naively reservoir sampling naively takes something like $O(nd + n_0 \log n_0)$ with high probability ... but will look it up tomorrow, unless someone else finds a reference first :-).}
	From Theorem~\ref{thm:dim_reduction}, the runtime of a single iteration of the loop is $O(n_0 d m) + T_{m,\eps/2}$, where $T_{m,\eps/2}$ is the runtime of solving the approximate mode finding problem in $m$ dimensions.
	In our case, $T_{m,\eps/2}$ consists of the runtime of the second coreset result as well as Algorithm~\ref{alg:sys_poly}.
	%The random project $\Pi$ takes time $O(d m n_0)$.  
%	\[
%	O(d m n_0)  
%	= 
%	O\left(d \left(\frac{1}{\eps^2}\log\frac{1}{\eps\rho\delta}\log^2\frac{1}{\eps\rho}\right) \left(\frac{1}{\eps^2\rho^2}\log\frac{1}{\delta}\right)\right) 
%	= 
%	O\left(\frac{d}{\eps^4 \rho^2} \log\frac{1}{\eps\rho\delta}\log^2\frac{1}{\eps\rho} \log \frac{1}{\delta}\right).
%  	\]
	From Section~\ref{sec:coresets}, it takes time $O(n_0 \mathrm{poly}(1/\eps \rho))$ to compute the second coreset $P^j_2$.  
    Then, Theorem~\ref{thm:ig-alg} implies that Algorithm \ref{alg:sys_poly} requires $O(n_2 \log n_2 \cdot (2 \sqrt{2 e \pi})^m + n_2 \cdot \left(\log \frac{m}{\eps\rho}\right)^{O(m)})$.  
    %Finally, the Mean-shift step takes $O(n_0 d)$ time.
    
    The single-loop runtime is dominated by the runtime of Algorithm~\ref{alg:sys_poly} (see Lemma~\ref{lem:poly-logc} for a simple proof).
    %After reading the input, this is dominated by the runtime of Algorithm \ref{alg:sys_poly} (see Lemma \ref{lem:poly-logc} to see this dominates $n_0 \mathrm{poly}(1/\eps\rho)$) which takes time
    Writing out the runtime of Algorithm~\ref{alg:sys_poly} gives
\begin{align*}
 \MoveEqLeft   O\left(n_2 \log n_2 \cdot (2 \sqrt{2 e \pi})^m + n_2 \cdot \left(\log \frac{m}{\eps\rho}\right)^{O(m)} \right)
 \\    &= 
 O\left(n_2 \cdot \left(\log \frac{m}{\eps\rho}\right)^{O(m)}\right)
 \\ &=   
    O\left( \frac{1}{\eps^2\rho}\log^2\frac{1}{\eps\rho} \cdot
                \left(\log \frac{1}{\eps^3 \rho} \log \frac{1}{\eps \rho } \log^2 \frac{1}{\eps \rho}\right)^{O(\frac{1}{\eps^2}\log^3\frac{1}{\eps\rho})}\right)  
\\ &=   
    O\left( \left(\log \frac{1}{\eps \rho } \right)^{O(\frac{1}{\eps^2}\log^3\frac{1}{\eps\rho})}\right).
\end{align*}

%We can repeat Algorithm \ref{alg:full} $O(\log \frac{1}{\delta})$ time and return the best solution.
Combining with the runtimes for reading the input, coreset construction and evaluating solutions in Step 9, then repeating the loop for $O(\log\frac{1}{\delta})$ times gives the bound in the theorem statement.  If $n < n_0 \log \frac{1}{\delta}$, we use the full set $P$ as each $P_0^j$.  
\end{proof}

In the regime where $\eps$ (the relative error) and $\delta$ (the probability of failure) are constant, the runtime simplifies to 
$O\left(\left(n+\frac{1}{\rho^2}\right)d + 
\left(\log \frac{1}{\rho}\right)^{O(\log^3\frac{1}{\rho})}
\right)$.
Note however that if $1/\rho^2 \le n_0$ dominates $n$, then we would not have constructed the coresets $P^j_0$ in the first place but used the entire point set instead, and so we can treat the first term as just $O(nd)$.
We also recall that $\rho = \Gb_P(x^*) \ge 1/n$, which by substitution gives an upper bound of $O\left(nd + 
\left(\log n\right)^{O(\log^3 n)}
\right)$.
\section{KDE mode finding for Two Dimensional Case}\label{sec:d=2}

In this subsection, we assume that $P\subset \mathbb{R}^2$ and $p=(p_1,p_2)$ for each $p\in  P$.  We can improve our low-dimensional analysis that used a set of systems of polynomials by about a logarithmic factor using a fairly different approach.  This shows how to approximate each Gaussian by a weighted set of rectangles.  After sampling by these weights, we can quickly retrieve the point of maximum depth in these rectangles as an approximation of the maximum.

We first define the following notation.
We let $s = \frac{\eps\rho}{6}$ be a minimal additive error we will allow for the spatial approximation, and then $m = \lceil\frac{1}{s}\rceil$ will be the number of discretizations we will need.  
A Gaussian has infinite support, but we will only need to consider $m$ such widths defined
$r_j=\sqrt{\log \frac{1}{l_j}}$
with $l_j = 1-\frac{j}{m}$
for $j=0,1,\ldots,m$.  As a special case we set $r_m = \infty$ (note that this allows $e^{-r_j^2} = l_j$).  
%$l_j = 1-\frac{j}{m}$ for $j=0,1,\dots,m$ and $r_j=\sqrt{\log \frac{1}{l_j}}$ for $j=0,1,\dots,m-1$ and $r_m=\infty$ (note that this allows $e^{-r_j^2} = l_j$).
We can now define a series of axis-parallel rectangles centered any a point $p = (p_1,p_2) \in P$ as
$\mathcal{R}_p = \setdef{[p_1-r_{a_1},p_1+r_{a_1}]\times[p_2-r_{a_2},p_2+r_{a_2}]}{(a_1,a_2)\in\{0,1,\dots,m-1\}^2}$.  
%Namely, $\mathcal{R}_p$ is a collection of axis-parallel rectangle centered at $p$.
It enumerates all possible widths $r_0,r_1,\dots,r_{m-1}$ on both directions and therefore its size is $m^2$.
Also, denote $\mathcal{R}$ be $\cup_{p\in P} \mathcal{R}_p$.

Given any $x\in\mathbb{R}^d$ and any finite collection $\mathcal{C}$ of subset of $\mathbb{R}^2$, denote $N(\mathcal{C},x)$ be the number of $C\in \mathcal{C}$ that $x\in C$.
That is, $N(\mathcal{C},x)$ is the depth or ply of $x$ with respect to $\mathcal{C}$.
And we can show that the normalized depth approximates the KDE value $\Gb_P(x)$. 

\begin{lemma}\label{lem:counting}
	\[
	\Gb_P(x)\geq \frac{N(\mathcal{R},x)}{n m^2} \geq \Gb_P(x) - \frac{1}{3}\eps \rho
	\]
\end{lemma}

We put the detail proof into Appendix \ref{app:proofs}. 
The main idea is to show that Gaussian kernel can be approximated by a collection of axis-parallel rectangle where $m$ indicates how precise it is.
Note that our $m$ does not show up in the actual algorithm since we will apply random sampling later.

Observe that $\abs{\mathcal{R}}=nm^2$.
We can rewrite Lemma~\ref{lem:counting} to be $\Gb_P(x)\geq \frac{N(\mathcal{R},x)}{\abs{\mathcal{R}}} \geq \Gb_P(x) - \frac{1}{3}\eps \rho$.

Now consider $(X,\mathcal{S})$ be a range space with VC dimension $\nu$.  
Given $\eps>0$ and $\alpha>0$, we call a subset $Z$ of $X$ a relative $(\alpha,\eps)$-approximation for $(X,\mathcal{S})$ if, for any $\tau\in \mathcal{S}$, $\abs{\frac{\abs{X\cap \tau}}{\abs{X}} - \frac{\abs{Z\cap \tau}}{\abs{Z}}} \leq \eps M$ when $M = \max\{\frac{\abs{X\cap\tau}}{\abs{X}}, \alpha\}$.  A random sample of size $O\left( \frac{1}{\eps^2\alpha}(\nu\log \frac{1}{\alpha} + \log \frac{1}{\delta}) \right)$ is an $(\alpha,\eps)$-approximation with probability at least $1-\delta$~\cite{har2011relative}.

%\begin{lemma}[\cite{har2011relative}]~\label{lem:approx_sample}
%	Suppose $(X,\mathcal{S})$ be a range space with VC dimension $\nu$.
%	A random sample of $X$ of size
%	\[
%	O\left( \frac{1}{\eps^2\alpha}(\nu\log \frac{1}{\alpha} + \log \frac{1}{\delta}) \right)
%	\]
%	is a relative $(\alpha,\eps)$-approximation for $(X,\mathcal{S})$ with probability $1-\delta$.
%\end{lemma}

The range space $(\mathbb{R}^2,\mathcal{B})$ where $\mathcal{B}$ is the set of all axis-parallel box in $\mathbb{R}^2$ has VC dimension $4$.  Thus its dual range space $(\mathcal{B},\mathcal{D})$ where $\mathcal{D}=\setdef{\setdef{B\in\mathcal{B}}{x\in B}}{x\in\mathbb{R}^2}$, has VC dimension is $O(1)$.
Denote $\mathcal{D}_0$ be $\setdef{\setdef{R\in\mathcal{R}}{x\in \mathcal{R}}}{x\in\mathbb{R}^2}$.  

Moreover given a set of $\lambda$ axis-aligned rectangles in $\R^2$, Chan~\cite{chan2013klee} shows how to find a maximal depth point in $O(\lambda \log \lambda)$ time.   
%
%\begin{lemma}[\cite{chan2013klee}]\label{lem:depth}
%	Given a collection $\mathcal{C}$ of $\lambda$ axis-parallel boxes in $\mathbb{R}^d$, we can find a point $x\in \mathbb{R}^d$ that maximizes the number of boxes in $\mathcal{C}$ containing $x$ in $O(\lambda\log \lambda)$ time for $d=2$ and in $O(\frac{\lambda^{d/2}}{\log^{d/2} \lambda}(\log\log \lambda)^{O(1)})$ time for $d\geq 3$.
%\end{lemma}
%
This leads to a simple algorithm for finding an approximate maximum point $x'$, in Algorithm \ref{alg:depth-d2}, with runtime stated in Theorem \ref{thm:max-d2}.  

\begin{algorithm}[H]
	\caption{\label{alg:depth-d2}Computing Depth}
	{\bf input}: a point set $P\subset \mathbb{R}^2$, parameter $\eps,\rho,\delta>0$
	%	{\bf output}: stationary point of $f(x) = \sum_{i=1}^n e^{-\norm{x - \mu_i}^2}$
	\begin{algorithmic}[1]
		\STATE generate a random subset $\mathcal{R}_0$ of $\mathcal{R}$ of size $O\left(\frac{1}{\eps^2\rho}(\log \frac{1}{\rho}+\log \frac{1}{\delta})\right)$.  % which is a relative $(\rho,\frac{1}{3}\eps)$-approximation for $(\mathcal{R},\mathcal{D}_0)$.  % by Lemma \ref{lem:approx_sample}.
		\STATE compute $x'\in\mathbb{R}^2$ such that $x'=\arg\max_{x\in\mathbb{R}^d} N(\mathcal{R}_0,x)$ using the algorithm by Chan~\cite{chan2013klee}.
		\STATE \textbf{return} $x'$
	\end{algorithmic}
%	\label{mean_shift}
\end{algorithm}

%Now, apply Lemma~\ref{approx_sample} to $(\mathcal{R},\mathcal{D}')$ with $\alpha=\rho$ to generate a subset $\mathcal{R}'$ of $\mathcal{R}$ of size $O\left(\frac{1}{\eps^2\rho}(\log \frac{1}{\rho}+\log \frac{1}{\delta})\right)$ which is a relative $(\rho,\eps)$-approximation for $(\mathcal{R},\mathcal{D}')$.
%Compute $x'\in\mathbb{R}^d$ such that $x'=\arg\max_{x\in\mathbb{R}^d} N(\mathcal{R}',x)$.

%\begin{observation}
%	$\frac{N(\mathcal{R},x^*)}{\abs{\mathcal{R}}}\geq (1-\eps)\rho$
%\end{observation}

\begin{theorem}[restated Theorem \ref{thm:2d}]\label{thm:max-d2}
	Given $\eps,\rho>0$ and a point set $P\subset\mathbb{R}^2$ of size $n$, we can find $x'\in\mathbb{R}^2$ so
	\[
	\Gb_P(x')\geq (1-\eps)\Gb_P(x^*)
	\]
	if $\Gb_P(x^*)\geq \rho$ where $x^*=\argmax_{x\in\mathbb{R}^d} \Gb_P(x)$ in $O\left(n + \frac{1}{\eps^2\rho}(\log \frac{1}{\rho}+\log \frac{1}{\delta})\log (\frac{1}{\eps\rho}\log \frac{1}{\delta})\right)$ time with probability at least $1-\delta$.
\end{theorem}

\begin{proof}
	First, by Lemma \ref{lem:counting}, $\frac{N(\mathcal{R},x^*)}{\abs{\mathcal{R}}}\geq \Gb_P(x^*) - \frac{1}{3}\eps\rho = \Omega(\rho)$.
	Let $M$ be $\max\{\frac{N(\mathcal{R},x')}{|\mathcal{R}|}, \rho\}$.
	We also have $M = \max\{\frac{N(\mathcal{R},x')}{|\mathcal{R}|}, \rho\} \leq \Gb_P(x^*)$.
	
	\begin{align*}
		\Gb_P(x')
		& \geq
		\frac{N(\mathcal{R},x')}{\abs{\mathcal{R}}} & \text{by Lemma \ref{lem:counting}}\\
		& \geq
		\frac{N(\mathcal{R}_0,x')}{\abs{\mathcal{R}_0}} - \frac{1}{3}\eps M   & \text{by the construction of $\mathcal{R}_0$}\\
		& \geq
		 \frac{N(\mathcal{R}_0,x^*)}{\abs{\mathcal{R}_0}} - \frac{1}{3}\eps M\\
		& \geq
		 \frac{(1-\frac{1}{3}\eps)N(\mathcal{R},x^*)}{\abs{\mathcal{R}}} - \frac{1}{3}\eps M   & \text{by $\frac{N(\mathcal{R},x^*)}{\abs{\mathcal{R}}}=\Omega(\rho)$ and construction of $\mathcal{R}_0$}\\
		& \geq
		(1-\frac{1}{3}\eps)(\Gb_P(x^*) - \frac{1}{3}\eps\rho) - \frac{1}{3}\eps M &\text{by Chan~\cite{chan2013klee}}\\
		& \geq
		(1-\eps)\Gb_P(x^*) & \text{since $M  \leq \Gb_P(x^*)$}
	\end{align*}
	
	To see the running time, note that the size of input $\lambda=O\left(\frac{1}{\eps^2\rho}(\log \frac{1}{\rho}+\log \frac{1}{\delta})\right)$ in our context and $O(n)$ time to create a sample.
	Therefore, the total running time is $O\left(n + \frac{1}{\eps^2\rho}(\log \frac{1}{\rho}+\log \frac{1}{\delta})\log (\frac{1}{\eps\rho}\log \frac{1}{\delta})\right)$.
\end{proof}

%%%%%%%%%%%%%%%%%%%%%%%%%%%%%%%%%%%%%%%%%%%%%%%%%%%%%%%%%%%%%%%%%%
%%%%%%%%%%%%%%%%%%%%%%%%%%%%%%%%%%%%%%%%%%%%%%%%%%%%%%%%%%%%%%%%%%
%%%%%%%%%%%%%%%%%%%%%%%%%%%%%%%%%%%%%%%%%%%%%%%%%%%%%%%%%%%%%%%%%%
%\section{Conclusion}
%This part is the conclusion.

\bibliography{main}
%\bibliographystyle{plain}

%We can move this to after the bibliography if we need space.  
\appendix

%%%%%%%%%%%%%%%%%%%%%%%%%%%%%%%%%%%%%%%%%%%%%%%%%%%%%%%%%%%%%%%%%%
%%%%%%%%%%%%%%%%%%%%%%%%%%%%%%%%%%%%%%%%%%%%%%%%%%%%%%%%%%%%%%%%%%
%%%%%%%%%%%%%%%%%%%%%%%%%%%%%%%%%%%%%%%%%%%%%%%%%%%%%%%%%%%%%%%%%%
%%%%%%%%%%%%%%%%%%%%%%%%%%%%%%%%%%%%%%%%%%%%%%%%%%%%%%%%%%%%%%%%%%

%%%%%%%%%%%%%%%%%%%%%%%%%%%%%%%%%%%%%%%%%%%%%%%%%%%%%%%%%%%%%%%%%%
%%%%%%%%%%%%%%%%%%%%%%%%%%%%%%%%%%%%%%%%%%%%%%%%%%%%%%%%%%%%%%%%%%
%%%%%%%%%%%%%%%%%%%%%%%%%%%%%%%%%%%%%%%%%%%%%%%%%%%%%%%%%%%%%%%%%%
\section{Technical Proofs}
\label{app:proofs}

\begin{lemma}[restated Lemma \ref{lem:sys_poly_guar}]
\label{app:sys_poly_guar}
	Suppose $r+r'>1$ and $q\in\R^d$ such that $\norm{x^*-q} \leq r\sqrt{\log \frac{1}{\eps\rho}}$.
	Then, the output $x^{(q)}$ of $\textsf{SysPoly}(P,q,r,r')$ satisfies
	\[
	\G_{Q_{P,q}(r')}(x^{(q)}) \geq \G_{Q_{P,q}(r')}(x^*) - \abs{Q_{P,q}(r')}\frac{\eps \rho}{2}
	\]
\end{lemma}

\begin{proof}
	
	First, we write $\sum_{p\in Q_q(r')} e^{-\norm{p-x^{(q)}}^2}$ into the following form.
	\begin{align*}
	\sum_{p\in Q_q(r')} e^{-\norm{p-x^{(q)}}^2}
	& =
	\sum_{p\in Q_q(r')} \prod_{i=1}^d \left(\sum_{j=0}^\infty\frac{1}{j!}\left(-(x^{(q)}_i-p_i)^2\right)^j\right) \\
	& =
	\sum_{p\in Q_q(r')} \prod_{i=1}^d \left(\sum_{j=0}^s\frac{1}{j!}\left(-(x^{(q)}_i-p_i)^2\right)^j\right) + \mathcal{E}(x^{(q)})
	\end{align*}
	where $\mathcal{E}(x) = \sum_{p\in Q_q(r')} \sum_{j_1,\dots,j_d\mid \text{one of $j_i\geq s$}}\frac{1}{j_1!\cdots j_d!}\left(-(x_1-p_1)^2\right)^{j_1}\cdots \left(-(x_d-p_d)^2\right)^{j_d}$ for any $x\in \mathbb{R}^d$.

	Now, we have
	\begin{align*}
	\sum_{p\in Q_q(r')}e^{-\norm{x^{(q)}-p}^2}
	& =
	\sum_{p\in Q_q(r')} \prod_{i=1}^d \left(\sum_{j=0}^s\frac{1}{j!}\left(-(x^{(q)}_i-p_i)^2\right)^j\right) + \mathcal{E}(x^{(q)}) \\
	& \geq
	\sum_{p\in Q_q(r')} \prod_{i=1}^d \left(\sum_{j=0}^s\frac{1}{j!}\left(-(x^*_i-p_i)^2\right)^j\right) - \abs{P}\frac{\eps\rho}{10} + \mathcal{E}(x^{(q)}) \\
	& \geq
	\sum_{p\in Q_q(r')}e^{-\norm{x^{*}-p}^2} - \abs{P}\frac{\eps\rho}{10}  + \mathcal{E}(x^{(q)}) - \mathcal{E}(x^*)
	\end{align*}
%	Note that $\norm{x^*-p} \geq \norm{q-p} - \norm{q-x^*} \geq (r')\sqrt{\log \frac{1}{\eps\rho}} - r\sqrt{\log \frac{1}{\eps\rho}} = \sqrt{\log \frac{1}{\eps\rho}}$ since $x^*\in B_q\left(r\sqrt{\log \frac{1}{\eps\rho}}\right)$ and $p\notin B_q\left((r')\sqrt{\log \frac{1}{\eps\rho}}\right)$ for $p\notin Q_q(r')$.
%	\begin{align*}
%	\sum_{p\in Q_q}e^{-\norm{x^{(q)}-p}^2}
%	& \geq
%	\sum_{p\in P} e^{-\norm{x^{*}-p}^2} - \sum_{p\notin Q_p} \eps\rho  + \mathcal{E}(x^{(q)}) - \mathcal{E}(x^*)\\
%	& \geq
%	\sum_{p\in P} e^{-\norm{x^{*}-p}^2} - n\eps\rho  + \mathcal{E}(x^{(q)}) - \mathcal{E}(x^*)
%	\end{align*}
	
	In order to analyze the term $\mathcal{E}(x^{(q)})$ and $\mathcal{E}(x^*)$, we can first analyze the term 
	\[\abs{\sum_{j_1,\dots,j_d\mid \text{one of $j_i\geq s$}}\frac{1}{j_1!\cdots j_d!}\alpha_1^{j_1}\cdots \alpha_d^{j_d}}\]
	 where $\alpha_i = -(y_i-p_i)^2$ where $y$ is $x^{(q)}$ or $x^*$.
	\begin{align*}
	\sum_{j_1,\dots,j_d\mid \text{one of $j_i\geq s$}}\left(\prod_{i=1}^d\frac{1}{j_i!}\alpha_i^{j_i}\right)
	& =
	\sum_{i=1}^d \left(\prod_{k=1}^{i-1}\sum_{j=0}^{s-1}\frac{1}{j!}\alpha_k^j\right)\left(\sum_{j=s}^\infty\frac{1}{j!}\alpha_i^j\right)\left(\prod_{k=i+1}^d\sum_{j=0}^\infty\frac{1}{j!}\alpha_k^j\right)
	\end{align*}
	For each $i=1,2,\dots,d$, by taking $s= (r+r')^2e^2\log \frac{d}{\eps\rho}$,
	\begin{align*}
	\abs{\sum_{j=s}^\infty \frac{1}{j!}\alpha_i^j }
	\leq
	\sum_{j=s}^\infty \frac{1}{j!}\abs{\alpha_i }^j 
	\leq
	\max_{\xi\in[-\abs{\alpha_i},\abs{\alpha_i}]}\frac{e^{\xi}}{s!}\abs{\alpha_i}^s 
	\end{align*}
	The last inequality is the error approximation of Taylor expansion of exponential function.
	Note that $\abs{\alpha_i} = (y_i-p_i)^2 \leq \norm{y-p}^2 \leq \left(\norm{y-q} + \norm{p-q}\right)^2 \leq \left(r\sqrt{\log\frac{1}{\eps\rho}} + r'\sqrt{\log\frac{1}{\eps\rho}}\right)^2 \leq (r+r')^2\log\frac{1}{\eps\rho}$.
	\begin{align*}
	\abs{\sum_{j=s}^\infty \frac{1}{j!}\alpha_i^j }
	& \leq
	\frac{e^{(r+r')^2\log\frac{1}{\eps\rho}}}{s!}((r+r')^2\log\frac{1}{\eps\rho})^s \\
	& \leq
	\frac{e^{(r+r')^2\log\frac{1}{\eps\rho}}}{s^s}((r+r')^2e\log\frac{1}{\eps\rho})^s &\text{by $s! \geq (\frac{s}{e})^s$} \\
	& \leq
	\frac{e^{(r+r')^2\log\frac{1}{\eps\rho}}}{e^{s}}
	\leq 
	(\frac{\eps\rho}{d})^{(r+r')^2(e^2-1)} &\text{recall that $s= (r+r')^2e^2\log \frac{d}{\eps\rho}$} \\
	& \leq
	\frac{\eps\rho}{20d} &\text{by $r+r'>1$ and for sufficient small $\eps\rho$} 
	\end{align*}
	
	Now, we can plug this into $\abs{\sum_{j_1,\dots,j_d\mid \text{one of $j_i\geq s$}}\frac{1}{j_1!\cdots j_d!}\alpha_1^{j_1}\cdots \alpha_d^{j_d}}$.
	\begin{align*}
	&
	\abs{\sum_{j_1,\dots,j_d\mid \text{one of $j_i\geq s$}}\frac{1}{j_1!\cdots j_d!}\alpha_1^{j_1}\cdots \alpha_d^{j_d}} \\
	& =
	\abs{\sum_{i=1}^d \left(\prod_{k=1}^{i-1}\sum_{j=0}^{s-1}\frac{1}{j!}\alpha_k^j\right)\left(\sum_{j=s}^\infty\frac{1}{j!}\alpha_i^j\right)\left(\prod_{k=i+1}^d\sum_{j=0}^\infty\frac{1}{j!}\alpha_k^j\right)} \\
	& \leq
	\sum_{i=1}^d \left(\prod_{k=1}^{i-1}(1+\frac{\eps\rho}{10d})\right)\left(\frac{\eps\rho}{10d}\right)\left(\prod_{k=i+1}^de^{\alpha_k}\right) \\
	& \leq
	\left(1+\frac{\eps\rho}{20d}\right)^d\frac{\eps\rho}{20} 
	\leq
	e^{\frac{\eps\rho}{20}}\frac{\eps\rho}{20}
	\leq
	\frac{\eps\rho}{8} & \text{for sufficient small $\eps\rho$}
	\end{align*}
	
	That means
	\begin{align*}
	\sum_{p\in Q_q(r')}e^{-\norm{x^{(q)}-p}^2}
	& \geq
	\sum_{p\in Q_q(r')}e^{-\norm{x^{*}-p}^2} - \abs{P}\frac{\eps\rho}{10}  + \mathcal{E}(x^{(q)}) - \mathcal{E}(x^*) \\ 
	& \geq
	\sum_{p\in Q_q(r')}e^{-\norm{x^{*}-p}^2} - \abs{P}\frac{\eps\rho}{10} - \abs{Q_{P,q}(r')}\frac{\eps\rho}{8} - \abs{Q_{P,q}(r')}\frac{\eps\rho}{8} \\
	& = 
	\sum_{p\in Q_q(r')}e^{-\norm{x^{*}-p}^2} - \abs{P}\frac{\eps\rho}{2} \qedhere
	\end{align*}
	
%	which means
%	\begin{align*}
%	\sum_{p\in Q_q}e^{-\norm{x^{(q)}-p}^2}
%	& \geq
%	\sum_{p\in P} e^{-\norm{x^{*}-p}^2} - \abs{P}\eps\rho  + \mathcal{E}(x^{(p)}) - \mathcal{E}(x^*)
%	\geq
%	\sum_{p\in P} e^{-\norm{x^{*}-p}^2} - \abs{P}O(\eps\rho)
%	\end{align*}
%	
%	Plugging this into $\G_P(x^{(q)})$,
%	\begin{align*}
%	\G_P(x^{(q)}) 
%	=
%	\sum_{p\in P} e^{-\norm{p-x^{(q)}}^2} 
%	\geq
%	\sum_{p\in Q_q} e^{-\norm{p-x^{(q)}}^2} 
%	\geq
%	\G_P(x^*) - O(\eps) \abs{P}\rho 		
%	\end{align*}
\end{proof}

\begin{lemma} \label{lem:poly-logc}
For $n \geq 1$ and $c \geq 1$ then $n^c \leq (\log n)^{c \log^3 n}$.
\end{lemma}
\begin{proof}
The proof is straightforward algebra (assuming the base of $\log$ is $2$)
\begin{align*}
(\log n)^{c \log^3 n}
&= 
2^{c \log^3 n \log \log n}
\\ & =
n^{c \log^2 n \log \log n}
\\ & \geq 
n^c. \qedhere
\end{align*}
\end{proof}

%\begin{proof}
%We start with $1 \leq \log^3(\log n) \cdot \log^6 (n^c)$.  We then obtain the proof with algebra
%\begin{align*}
%1 & \leq \log^3(\log n) \cdot \log^6 (n^c)
%\\
%\frac{1}{\log^3(n^c)} \cdot \log^3(n^c)    & \leq   \log^3(\log n) \cdot \log^6 (n^c)
%\\
%\frac{1}{\log^9(n^c)} \cdot \log^3(n^c)    & \leq   \log^3(\log n) 
%\\
%\left(\frac{1}{\log^3(n^c)} \cdot \log(n^c)\right)^3    & \leq   \log^3(\log n) 
%\\
%\log^3\left((n^c)^{\frac{1}{\log^3(n^c)}}\right)    & \leq   \log^3(\log n) 
%\\
%(n^c)^{\frac{1}{\log^3(n^c)}}   & \leq   \log n
%\\
%n^c   & \leq   (\log n)^{\log^3(n^c)}
%\\
%n^c   & \leq   (\log n)^{c^3\log^3 n}  \qedhere
%\end{align*}
%\end{proof}
This implies that $\mathrm{poly}(n) = (\log n)^{O(\log^3)}$.  A relevant context is where $m = O((1/\eps^2) \log^3(1/\eps \rho))$ so $\mathsf{poly}(1/\eps \rho) \leq (\log \frac{m}{\eps \rho})^{O(m)}$.  

\begin{lemma}[restated Lemma \ref{lem:counting}]\label{app:counting}
	\[
	\Gb_P(x)\geq \frac{N(\mathcal{R},x)}{n m^2} \geq \Gb_P(x) - \frac{1}{3}\eps \rho
	\]
\end{lemma}

\begin{proof}	
	For any $p\in P$ and $i\in \{1,2\}$, let $a_i$ be the integer such that $r_{a_i-1} \leq \abs{p_i-x_i} \leq r_{a_i}$ which implies $e^{-r_{a_i-1}^2} \geq e^{-(p_i-x_i)^2} \geq e^{-r_{a_i}^2} = 1-\frac{a_i}{m}$.
	Then, we have 
	\[
	e^{-\norm{p-x}^2} = e^{-(p_1-x_1)^2-(p_2-x_2)^2} \geq (1-\frac{a_1}{m})(1-\frac{a_2}{m}) = \frac{N(\mathcal{R}_p,x)}{m^2}
	\]
	Note that $N(\mathcal{R},x) =\sum_{i=1}^d N(\mathcal{R}_p,x)$.
	Now, 
	\begin{align*}
	\G_P(x)
	=
	\sum_{p\in P} e^{-\norm{p-x}^2}
	\geq
	\sum_{p\in P} \frac{N(\mathcal{R}_p,x)}{m^2}
	=
	\frac{N(\mathcal{R},x)}{m^2}		
	\end{align*}
	
	On the other hand, let $\Delta_{p,x,i}=e^{-(p_i-x_i)^2} - (1-\frac{a_i}{m})$ which is larger than $0$,
	\begin{align*}
	\frac{N(\mathcal{R}_p,x)}{m^2}
	& =
	(1-\frac{a_1}{m})(1-\frac{a_2}{m})
	=
	\left(  e^{-(p_1-x_1)^2} - \Delta_{p,x,1} \right)\left(  e^{-(p_2-x_2)^2} - \Delta_{p,x,2} \right) \\
	& = 
	e^{-(p_1-x_1)^2}e^{-(p_2-x_2)^2} - \Delta_{p,x,1}e^{-(p_2-x_2)^2} - \Delta_{p,x,2}e^{-(p_1-x_1)^2} + \Delta_{p,x,1}\Delta_{p,x,2} \\
	& \geq
	e^{-(p_1-x_1)^2}e^{-(p_2-x_2)^2} - \Delta_{p,x,1}e^{-(p_2-x_2)^2} - \Delta_{p,x,2}e^{-(p_1-x_1)^2}
	\end{align*}
	Recall that $e^{-r_{a_i-1}^2} \geq e^{-(p_i-x_i)^2} \geq e^{-r_{a_i}^2}$ which implies $\Delta_{p,x,i} \leq e^{-r_{a_i-1}^2} - e^{-r_{a_i}^2} = s$.
	The above equation becomes
	\begin{align*}
	\frac{N(\mathcal{R}_p,x)}{m^2}
	& \geq 
	e^{-(p_1-x_1)^2}e^{-(p_2-x_2)^2} - \Delta_{p,x,1}e^{-(p_2-x_2)^2} - \Delta_{p,x,2}e^{-(p_1-x_1)^2} \\
	& \geq
	e^{-\norm{p-x}^2} - 2s
	\end{align*}
	Finally, we have
	\begin{align*}
	\frac{N(\mathcal{R},x)}{m^2}
	=
	\sum_{p\in P}\frac{N(\mathcal{R}_p,x)}{m^2} 
	\geq 
	\sum_{p\in P}(e^{-\norm{p-x}^2} - 2s) 
	=
	\G_P(x) - \frac{1}{3}\eps n\rho. \phantom{1234134512dggsdfg545} \qedhere
	\end{align*}	
\end{proof}

\end{document}